\documentclass[11pt,a4paper]{article}
\newlength{\frase}

\usepackage{amsmath,amssymb,amsthm}
\usepackage{stmaryrd}

\usepackage{graphics}
\usepackage{pdfsync} 
\usepackage{color}
\newcommand{\rosso}{}          % it could be removed for the final version
\newcommand{\green}{}     % it could be removed for the final version
\newcommand{\blue}{}        % it could be removed for the final version
\theoremstyle{plain}
     \newtheorem{theorem}{Theorem}[section]
     \newtheorem{lemma}[theorem]{Lemma}
     \newtheorem{corollary}[theorem]{Corollary}
     \newtheorem{proposition}[theorem]{Proposition}
     
\theoremstyle{definition}
     \newtheorem{definition}[theorem]{Definition}

\theoremstyle{remark}
     \newtheorem{remark}{Remark}[section]
\mathcode`\<="4268  % < = \langle
\mathcode`\>="5269  % > = \rangle
\mathchardef\semicolon="603B % the original
\mathchardef\gt="313E
\mathchardef\lt="313C

\def\emptyset{\varnothing}
\def\epsilon{\varepsilon}
\def\theta{\vartheta}

\newcount\mscount
\newdimen\proofrulebreadth \proofrulebreadth=.05em
\newdimen\proofdotseparation \proofdotseparation=1.25ex
\newdimen\proofrulebaseline \proofrulebaseline=2ex
\newcount\proofdotnumber \proofdotnumber=3
\let\then\relax
\def\hfi{\hskip0pt plus.0001fil}
\mathchardef\squigto="3A3B
\newif\ifinsideprooftree\insideprooftreefalse
\newif\ifonleftofproofrule\onleftofproofrulefalse
\newif\ifproofdots\proofdotsfalse
\newif\ifdoubleproof\doubleprooffalse
\let\wereinproofbit\relax
\newdimen\shortenproofleft
\newdimen\shortenproofright
\newdimen\proofbelowshift
\newbox\proofabove
\newbox\proofbelow
\newbox\proofrulename
\def\shiftproofbelow{\let\next\relax\afterassignment\setshiftproofbelow\dimen0 }
\def\shiftproofbelowneg{\def\next{\multiply\dimen0 by-1 }%
\afterassignment\setshiftproofbelow\dimen0 }
\def\setshiftproofbelow{\next\proofbelowshift=\dimen0 }
\def\setproofrulebreadth{\proofrulebreadth}

\def\prooftree{% NESTED ZERO (\ifonleftofproofrule)
\ifnum	\lastpenalty=1
\then	\unpenalty
\else	\onleftofproofrulefalse
\fi
\ifonleftofproofrule
\else	\ifinsideprooftree
	\then	\hskip.5em plus1fil
	\fi
\fi
\bgroup% NESTED ONE (\proofbelow, \proofrulename, \proofabove,
\setbox\proofbelow=\hbox{}\setbox\proofrulename=\hbox{}%
\let\justifies\proofover\let\leadsto\proofoverdots\let\Justifies\proofoverdbl
\let\using\proofusing\let\[\prooftree
\ifinsideprooftree\let\]\endprooftree\fi
\proofdotsfalse\doubleprooffalse
\let\thickness\setproofrulebreadth
\let\shiftright\shiftproofbelow \let\shift\shiftproofbelow
\let\shiftleft\shiftproofbelowneg
\let\ifwasinsideprooftree\ifinsideprooftree
\insideprooftreetrue
\setbox\proofabove=\hbox\bgroup$\displaystyle % NESTED TWO
\let\wereinproofbit\prooftree
\shortenproofleft=0pt \shortenproofright=0pt \proofbelowshift=0pt
\onleftofproofruletrue\penalty1
}

\def\eproofbit{% NESTED TWO
\ifx	\wereinproofbit\prooftree
\then	\ifcase	\lastpenalty
	\then	\shortenproofright=0pt	% 0: some other object, no indentation
	\or	\unpenalty\hfil		% 1: empty hypotheses, just glue
	\or	\unpenalty\unskip	% 2: just had a tree, remove glue
	\else	\shortenproofright=0pt	% eh?
	\fi
\fi
\global\dimen0=\shortenproofleft
\global\dimen1=\shortenproofright
\global\dimen2=\proofrulebreadth
\global\dimen3=\proofbelowshift
\global\dimen4=\proofdotseparation
\global\mscount=\proofdotnumber
$\egroup  % NESTED ONE
\shortenproofleft=\dimen0
\shortenproofright=\dimen1
\proofrulebreadth=\dimen2
\proofbelowshift=\dimen3
\proofdotseparation=\dimen4
\proofdotnumber=\mscount
}

\def\proofover{% NESTED TWO
\eproofbit % NESTED ONE
\setbox\proofbelow=\hbox\bgroup % NESTED TWO
\let\wereinproofbit\proofover
$\displaystyle
}%
\def\proofoverdbl{% NESTED TWO
\eproofbit % NESTED ONE
\doubleprooftrue
\setbox\proofbelow=\hbox\bgroup % NESTED TWO
\let\wereinproofbit\proofoverdbl
$\displaystyle
}%
\def\proofoverdots{% NESTED TWO
\eproofbit % NESTED ONE
\proofdotstrue
\setbox\proofbelow=\hbox\bgroup % NESTED TWO
\let\wereinproofbit\proofoverdots
$\displaystyle
}%
\def\proofusing{% NESTED TWO
\eproofbit % NESTED ONE
\setbox\proofrulename=\hbox\bgroup % NESTED TWO
\let\wereinproofbit\proofusing
\kern0.3em$
}

\def\endprooftree{% NESTED TWO
\eproofbit % NESTED ONE
  \dimen5 =0pt% spread of hypotheses
\dimen0=\wd\proofabove \advance\dimen0-\shortenproofleft
\advance\dimen0-\shortenproofright
\dimen1=.5\dimen0 \advance\dimen1-.5\wd\proofbelow
\dimen4=\dimen1
\advance\dimen1\proofbelowshift \advance\dimen4-\proofbelowshift
\ifdim	\dimen1<0pt
\then	\advance\shortenproofleft\dimen1
	\advance\dimen0-\dimen1
	\dimen1=0pt
	\ifdim  \shortenproofleft<0pt
        \then   \setbox\proofabove=\hbox{%
			\kern-\shortenproofleft\unhbox\proofabove}%
                \shortenproofleft=0pt
        \fi
\fi
\ifdim	\dimen4<0pt
\then	\advance\shortenproofright\dimen4
	\advance\dimen0-\dimen4
	\dimen4=0pt
\fi
\ifdim	\shortenproofright<\wd\proofrulename
\then	\shortenproofright=\wd\proofrulename
\fi
\dimen2=\shortenproofleft \advance\dimen2 by\dimen1
\dimen3=\shortenproofright\advance\dimen3 by\dimen4
\ifproofdots
\then
	\dimen6=\shortenproofleft \advance\dimen6 .5\dimen0
	\setbox1=\vbox to\proofdotseparation{\vss\hbox{$\cdot$}\vss}
	\setbox0=\hbox{%
		\kern\dimen6
		$\vcenter to\proofdotnumber\proofdotseparation
			{\leaders\box1\vfill}$%
		\unhbox\proofrulename}%
\else	\dimen6=\fontdimen22\the\textfont2 % height of maths axis
	\dimen7=\dimen6
	\advance\dimen6by.5\proofrulebreadth
	\advance\dimen7by-.5\proofrulebreadth
	\setbox0=\hbox{%
		\kern\shortenproofleft
		\ifdoubleproof
		\then	\hbox to\dimen0{%
			$\mathsurround0pt\mathord=\mkern-6mu%
			\cleaders\hbox{$\mkern-2mu=\mkern-2mu$}\hfill
			\mkern-6mu\mathord=$}%
		\else	\vrule height\dimen6 depth-\dimen7 width\dimen0
		\fi
		\unhbox\proofrulename}%
	\ht0=\dimen6 \dp0=-\dimen7
\fi
\let\doll\relax
\ifwasinsideprooftree
\then	\let\VBOX\vbox
\else	\ifmmode\else$\let\doll=$\fi
	\let\VBOX\vcenter
\fi
\VBOX	{\baselineskip\proofrulebaseline \lineskip.2ex
	\expandafter\lineskiplimit\ifproofdots0ex\else-0.6ex\fi
	\hbox	spread\dimen5	{\hfi\unhbox\proofabove\hfi}%
	\hbox{\box0}%
	\hbox	{\kern\dimen2 \box\proofbelow}}\doll%
\global\dimen2=\dimen2
\global\dimen3=\dimen3
\egroup % NESTED ZERO
\ifonleftofproofrule
\then	\shortenproofleft=\dimen2
\fi
\shortenproofright=\dimen3
\onleftofproofrulefalse
\ifinsideprooftree
\then	\hskip.5em plus 1fil \penalty2
\fi
}

\newcommand{\urule}[3]{%
   \prooftree #1 \justifies #2 \using #3 \endprooftree}
\newcommand{\brule}[4]{%
   \prooftree #1\ \ \ #2 \justifies #3 \using #4 \endprooftree}

\newcommand{\trule}[5]{%
   \prooftree #1\ \ \ #2\ \ \ #3  \justifies #4 \using #5 \endprooftree}

\newcommand{\prova}[3]{%
   \prooftree #1 \justifies  #2  \using #3 \endprooftree}

\newcommand{\LT}[2]{%
   \prooftree #2 \leadsto #1 \endprooftree}

\newcommand{\DT}[3]{%
\settowidth{\frase}{$#2$}
   \prooftree #1 \using {\hspace{-2.5ex} 
\parbox[c][1.5em]{\frase}{$#2$}} %
    \proofdotseparation=1.2ex \proofdotnumber=0 \leadsto #3 \endprooftree}

\newcommand{\ext}[1]{\mathbf{E}{(#1)}}

\def\conc{}

\def\dueseq{\textsf{2--Sequents}}
\def\linest{\textsf{LNS}}

\def\KK{{\textsf{K}}}
\def\T{{\textsf{T}}}
\def\K4{{\textsf{K4}}}
\def\D{{\textsf{D}}}
\def\S4{{\textsf{S4}}}
\def\DQ{{\textsf{D4}}}  % added for revised version

\def\Q2{{\textsf{S4.2}}}

\def\nKK{{\mathcal{N}_{\textsf{K}}}}
\def\nT{\mathcal{N}_{{\textsf{T}}}}
\def\nK4{\mathcal{N}_{{\textsf{K4}}}}
\def\nD{\mathcal{N}_{{\textsf{D}}}}
\def\nDQ{\mathcal{N}_{{\textsf{D4}}}} % added for revised version
\def\nS4{\mathcal{N}_{{\textsf{S4}}}}
\newcommand{\ns}[1]{\mathcal{N}_{#1}}

\def\iKK{{\mathcal{N}^i_{\textsf{K}}}}
\def\iT{\mathcal{N}^i_{{\textsf{T}}}}
\def\iK4{\mathcal{N}^i_{{\textsf{K4}}}}
\def\iD{\mathcal{N}^i_{{\textsf{D}}}}
\def\iDQ{\mathcal{N}^i_{{\textsf{D4}}}} % added for revised version
\def\iS4{\mathcal{N}^i_{{\textsf{S4}}}}
\def\class{{\mathcal{N}}}
\def\int{{\mathcal{N}^i}}

\def\red{\mathbf{\succ}}

\def\rred{\stackrel{*}{\red}}

\newcommand\pf[2]{{#1}^{#2}}

\newcommand\ded[3]{%
\begin{array}{c}#1\\{#3}\end{array}%
\hspace{-1ex}{\mbox{\small $#2$}}}

\newcommand{\rep}[2]{[#1 \Rsh  #2]}

\newcommand{\rl}{\mathcal{R}}

\def\conc{}

\def\toc{\mathcal{T}}
\def\pos{\toc^*}

\def\NN{\mathbb{N}}

\def\NN{\mathcal{N}}

\newcommand{\less}[1]{\sqsubseteq_{#1}}
\newcommand{\lest}[1]{\sqsubset_{#1}}
\newcommand{\suc}[1]{\triangleleft_{#1}}

\newcommand{\iniz}[1]{\mathfrak{Init}[#1]}

\newcommand{\mdl}[1]{\mathcal{M}_{#1}}

\newcommand{\tmdl}[1]{\mathfrak{S}_{#1}}

\def\tree{\Theta}

\def\sys{\mathbb{M}}

\def\mform{\mathfrak{mf}}
\def\pform{\mathfrak{pf}}

\def\R{\mathfrak{R}}

\newcommand{\logm}[1]{\mathfrak{M}[#1]}

\newcommand{\srho}[1]{\ell_{{[#1]}}}

\newcommand{\vertil}[1]{\begin{array}{l} #1\end{array}}

\newcommand{\ptol}[1]{\llbracket #1 \rrbracket}
\title{From 2--sequents and Linear Nested Sequents to Natural Deduction for Normal Modal Logics}

\author{Simone Martini$^1$\qquad Andrea Masini$^2$\qquad Margherita Zorzi$^2$}
\date{\small $^1$Universit\`a di Bologna, and INRIA Sophia-Antipolis\\
$^2$Universit\`a di Verona\\
\today}

\begin{document}
\maketitle
\pagestyle{myheadings}
%\markboth{modal proof theory}{}

\begin{quotation}
   \small \noindent {\sc Abstract:} %
We extend to natural deduction the approach of Linear Nested Sequents and of 2-Sequents. Formulas are 
decorated with a spatial coordinate, which allows a formulation of formal systems in the original spirit of natural deduction---only one introduction and one elimination rule per connective, no additional (structural) rule, no  explicit reference to the accessibility relation of the intended Kripke models.
We give systems for the normal modal logics from \textsf{K} to \textsf{S4}. For the intuitionistic versions of the systems, we define proof reduction, and prove proof normalization, thus obtaining a syntactical proof of consistency. For logics \textsf{K} and \textsf{K4} we use existence predicates (\`a la Scott) for formulating sound deduction rules. \\

\noindent\small{To appear into \it ACM Transactions on Computational Logic}, 2021.
\end{quotation}
%%%%%%%%%%%%%%%%%%%%%%%%%%%
\thanks{
\small{\it ACM CCS Concepts}\/: Theory of computation $\rightarrow$ Proof theory; Modal and temporal logics.\\
\small{\it ACM Computing Classification System (1998)}:\/ F.4.1.}
%
%
%%%%%%%%%%%%%%%%%%%%%%%%%%%%

% !TEX root = ./NatDedPos-revised.tex

\section{Introduction}\label{sec:introduction}

Proof theory of modal logics is a  subtle subject, and if a sequent calculus presentation is complex, natural deduction systems are even more daunting. The source of the problem is already well highlighted in Dag Prawitz's foundational book~\cite{Prawitz:1965}.
 
One of the most successful proof-theoretical formulations of modal logics are the \textit{labelled systems} of~\cite{Vigano00a,Simpson93,Negri:2011},
which extend ordinary natural deduction by explicitly mirroring in 
the deductive apparatus the accessibility relation of Kripke models (see also~\cite{MVVJANCL2101,MVVJLC2011,MasiniViganoZorzi08,MVVENTCS2010,MVZ-jmvl, Baratella2019,Baratella2004a,BaMaJANCL13}). In a sense, they may look like a formalization of Kripke semantics in a
first-order deductive fashion (see Section~\ref{sect:comparisons}, below, for a more complete discussion).
 
Differently from the labelled systems cited above, we aim to define \emph{natural deduction systems} for modal logics \textit{that do not explicitly deal with the accessibility relation}. 
Our leading idea is to extend \emph{geometrically} the standard natural deductive systems for classical and intuitionistic logic, to treat modalities as quantifiers are treated in first-order systems. 
In doing this we refine and extend to natural deduction some recent proposals by Lellmann and others for sequent calculi for modal logics \cite{Pimentel2019,Lellmann2019} (see later in this introduction).

\subsubsection*{Our proposal in a nutshell}
We add to formulas a kind of spatial coordinates, that we call \textit{positions}, to adapt to natural deduction the paradigm of \dueseq{} by Masini~\cite{Mas:TwoSeqProof:92}, and of \rosso{Linear Nested Sequents (\linest{}, from now on)  by Lellmann~\cite{Lellmann2015}.}
The main features of our systems are the following:
\begin{itemize}
	\item[--] there is exactly one
	introduction and one elimination rule for each modal connective;
	\item[--] rules for modal connectives have the same shape as those of first order
	quantifiers;
	\item[--] no formalization of the first order translation of modal logic
	formulas is present at the level of deduction rules (hence no formalization of
	the accessibility relation appears);
	\item[--] a notion of \textit{proof reduction} is given  and \textit{normalization} is proved, following the standard definitions and techniques for natural deduction systems;
	\item[--]only modal operators  can  change the spatial positions of formulas.
\end{itemize}

We stress that, as was the case for \dueseq{} and \linest{}, a specific goal is not to explicitly embed  the notion of accessibility relation, thus equipping the formal systems with ad-hoc deductive rules (see also Section~\ref{sect:comparisons}).

\subsubsection*{A short history}
To fully understand our proposal it is useful to frame it ``historically'', and to go back to \dueseq{}, originally formulated in \cite{Mas:TwoSeqInt:93,Mas:TwoSeqProof:92}.  There, the main idea was to add a second dimension to ordinary propositional sequents. Each formula in a \textsf{2-Sequent} lives at a \textit{level} (that could be seen as a natural number). 

Such a proposal was later extended and generalized to a natural deduction setting. Formulas become indexed formulas, i.e. pairs of formulas and natural numbers, where numbers correspond explicitly  to  levels in  \dueseq{}.
Such an idea works fine for the negative  $\bot$-free fragments of  the modal logics 
\KK, \T, \K4 and \S4, and for the corresponding \textbf{MELL} (Multiplicative Exponential Linear Logic) subsystems \cite{MM:ComInt:95,MM:ONFineStr:95}.
At the time we presented such systems, however, it was not possible to extend them to  full modal logics from \KK{} to \S4, since the simple notion of level of a formula does not interact well with reduction when there are also  $\Diamond$ rules.

The problem  does not show up if, instead of natural deduction, we consider \dueseq{}---see e.g.~\cite{GueMarMasPNGC2001,GuMarMasCoher2003, GMM:an-exp-ext-seq} where the authors show  how \dueseq{} are a suitable framework to deal with full MELL (and other linear systems) both in sequent calculi, and proof nets. 

More recently, the approach based on \dueseq{} has been  extended to deal with linear and branching time temporal logics~\cite{BaMaJANCL13,BarMas:apal}.
In particular, for temporal logics it was necessary to properly extend the notion of level since natural numbers do not suffice.

Finally, the paradigm of \dueseq{} has been reformulated  by Lellmann and coauthors, under the name  of \linest~\cite{Lellmann2015, Lellmann2019,Pimentel2019}, to deal with a more interesting class of logics. 

Unfortunately, \dueseq/\linest{} cannot be directly translated into a natural deduction setting, since the simple decoration of formulas with natural numbers does not agree with the obvious definition of reduction. To overcome  these problems, the simple (simplistic) notion of level has to be generalized to that of \emph{position}. 

\subsubsection*{Content of the paper}
The paper deals with the normal modal logics varying from \KK{} to \S4. \blue{We start with the classical systems since they are the ``standard'' in the modal logic literature. We give systems for each logic, proving soundness and completeness with respect to the axiomatic formulation, passing through a suitable Kripke style semantics of our systems. We then focus on the intuitionistic fragments---obtained syntactically, as usual, by removing the \emph{reduction ab absurdum} rule. For the intuitionistic systems, we define a notion of reduction for proofs and we give a syntactical proof of normalization, along the lines of the analogous proof for standard natural deduction. This allows us to obtain a purely syntactic proof of consistency---as a by-product of normalization---which applies also to the classical systems, via a double-negation translation.}
\rosso{We conclude with a detailed discussion of the relations between our systems and the labelled ones (for \emph{modal natural deduction}), and with some considerations about obtained results and future work.}

\subsubsection*{On classical systems}
One may wonder why dealing with classical logics at all, if the specific results we prove for them are, in the end, only soundness and completeness. 
Proof theory is (and has always been) a way to expound the meaning of logical connectives, independently of a set-theoretic (Tarskian or Kripkean) semantics. This is especially true for natural deduction, through rules of introduction/elimination of a \emph{single} modal connective. 
Our rules for $\Box$ and $\Diamond$, thus, reveal the links that these connectives have to the quantifiers, well before, and independently, of their interpretation as quantifiers on nodes of a \emph{classical} Kripke structure. 

The fact that we prove normalization only for the intuitionistic systems does not mean that normalization does not hold for the classical ones. Only, consequences of normalization (e.g., subformula property) will hold only partially, or only for subsystems (e.g., dealing only with $\bot$, $\to, \wedge$, and $\Box$), as it happens for the first-order classical case (see again, as the only reference among the dozen possible, Prawitz's monograph~\cite[Chapter III]{Prawitz:1965}.)

Moreover, proof assistants are more and more important in computational logic. Natural deduction formulation of \emph{classical} modal logics (that is, the ones at the basis of the logics used in the specification and verification of computer systems) opens up new avenues in the field of mechanical reasoning for such systems. 

Finally, let us remark once more the interest of having a completely syntactic proof of consistency for classical modal logics, independent of the existence of a Kripke model. This should be especially dear to computational logicians, whose bread and butter is, indeed, syntax only.

% !TEX root = ./NatDedPos-revised.tex

\section{Preliminary Notions}\label{sect:sequences}

As mentioned in the introduction, formula occurrences will be labeled with \emph{positions}---sequences of uninterpreted \emph{tokens}. We introduce here the notation and operations that will be needed for such notions.

Given a set $X$,  $X^*$ is the set of ordered finite sequences on $X$. With $<x_1,...,x_n>$ we denote the finite non empty sequence s.t. $x_1,\ldots,x_n \in X$; $<\ >$  is the empty sequence.

The  (associative) concatenation  of sequences $\conc:X^*\times X^*\to X^*$ is defined as 
\begin{itemize}
	\item $<x_1,...,x_n>\conc <z_1,...,z_m> = <x_1,...,x_n,z_1,...,z_m>$,
	\item $s\conc<\ > =<\ >\conc s= s$.
\end{itemize}

For $s\in X^*$ and $x\in X$, we sometimes write $s\conc x$ for $s\conc <x>$; and $x\in s$ as a shorthand for $\exists t,u\in X^*.\; s=t\conc<x>\conc u$.
The set $X^*$ is equipped with the following successor  relation

\[ s \suc{X} t \Leftrightarrow \exists x\in X.\; t=s\conc <x>\]

We use the following notations:
\begin{itemize}
	\item $\suc{X}^0$  denotes the reflexive closure of $\suc{X}$;
	\item $\lest{X}$  denotes the transitive closure of $\suc{X}$;
	\item $\less{X}$  denotes the reflexive and transitive closure of $\suc{X}$;
\end{itemize}

Given three sequences $s,u,v\in X^*$ the \emph{prefix replacement} $s\rep{u}{v}$ is so defined 

$
s\rep{u}{v}=
\begin{cases}
v\conc t \quad \mbox{if\ } s=u\conc t
\\
s \quad\mbox{otherwise}
\end{cases}
$

When $u$ and $v$ have the same  length, the replacement is called \textit{renaming} of $u$ with $v$.

%%%%%%%%%%%%%%%%%%%%%%%%%%%%%%%%%%%%%%%%%%%%%%%%%%%%%%%%%%
%%%%%%%%%%%%%%%%%%%%%%%%%%%%%%%%%%%%%%%%%%%%%%%%%%%%%%%%%%
\section{Modal Languages and Systems}\label{sec:2seq}

The propositional modal language
${\mathcal L}$ contains the following symbols:
\begin{enumerate}
	\item[--]  
	countably infinite \textit{proposition symbols}, $p_0,p_1,\ldots$;  %from a countably infinite set $\prop$
	\item[--] the \textit{propositional connectives} $\lor, \land, \to, \bot;$
	\item[--]the \textit{modal operators} $\Box, \Diamond;$
	\item[--]the \textit{auxiliary symbols} $($\ and $).$
\end{enumerate}

As usual, $\neg A$ is a shorthand for $A\to \bot$.

\begin{definition} The set $\mform$ of propositional \textit{modal formulas} 
	of ${\mathcal L}$ is the least set that contains the propositional
	symbols and is closed under application of the propositional
	connectives and the modal operators.
{A formula is atomic if it is a propositional symbol, or the connective $\bot$.}
\end{definition}

In the following, $\toc$ denotes a denumerable set of \emph{tokens}, ranged by meta-variables $x,y,z$, possibly indexed.
Let $\pos$ be the set of the sequences on $\toc$, called \textit{positions}; meta-variables $\alpha,\beta,\gamma$  range on $\pos$, possibly indexed.

\begin{definition}\label{def:levfor} 
A \textit{position-formula} (briefly \emph{p-formula}) is an expression of the form $\pf{A}{\alpha}$, where
		$A$ is a modal formula and $\alpha\in \pos$. We denote by $\pform$ the set of position formulas.
\end{definition}
Given a sequence $\Gamma$ of p-formulas,  $\iniz{\Gamma}$  is the set of prefixes of the positions in $\Gamma$:\\
$\{ \beta : \exists  \pf{A}{\alpha} \in \Gamma.\; \beta \less{} \alpha \}$.

It could be useful to anticipate that, in the semantics we will define in Section~\ref{subs:seman}, positions will be mapped into nodes of a Kripke structure (and hence sublists of a position will range on paths of nodes). Affirming $\pf{A}{\alpha}$ in a Kripke model $\mdl{}$, means that $A$ is true at $\alpha$ in $\mdl{}$.  We stress, however, that positions are, at this point, a mere technical proof-theoretical device, whose aim is to mimic as much as possible the behaviour of first order variables in standard natural deduction. Under this informal interpretation, $\pf{A}{\alpha}$ could be seen as a formula with its free variables in $\alpha$. The modal introduction rules (which work as the quantifier ones in standard natural deduction) act on the position of their main premise, removing (``binding'') some of the tokens of the position. Analogously,  elimination rules allow some form of ``instantiation'' on positions. The possibility to work on sublists of positions is the key ingredient of our approach, when compared to labelled systems, where labels must be treated one-by-one.

%%%%%%%%%%%%%%%%%%%%%%%%%%%%%%%%%%%%%%%%%%%%%%%%%%%%%%%%%
\subsection{A  class of normal modal systems}
We briefly recall  the axiomatic (``Hilbert-style'') presentation of normal modal systems.
Let $Z$ be a set of formulas. 
The normal modal logic $\logm{Z}$ is defined as smallest set $X$ of formulas verifying the following properties: 

\begin{description}
	\item[(i)] $Z\subseteq X$
	\item[(ii)] $X$ contains all instances of the following schemas:
	\begin{description}
		\item[1.]
		$A\to(B \to A)$
		\item[2.] $(A\to (B\to C))\to ((A\to B)\to (A\to C))$
		\item[3.] $((\neg B\to \neg A)\to ((\neg  B\to A)\to B)) $
		\item[K.] $\Box(A\to B)\to(\Box A \to \Box B)$
	\end{description}
	\item[MP]  if $A,A\to B \in X$ then $B\in X$;
	\item[NEC] if $A\in X$ then $\Box A\in X$. 
\end{description}

We write $\vdash_{\logm{Z}}  A$  for  $A\in \logm{Z}$.
If $N_1,..,N_k$ are names of schemas, the sequence $N_1\ldots N_k$ denotes the set  
$[N_1]\cup...\cup [N_1]$, where  
$[N_i] =\{A: A \mbox{\ is an instance of the schema\ } N_i\}$.
Figure \ref{fig:modlogics} lists the standard axioms for the well-known modal systems ${\KK}$,  ${\D}$,  ${\T}$,  ${\K4}$,  \rosso{\DQ}, ${\S4}$; we use $\sys$ as a generic name for one of these systems.

\begin{figure}
	\center
	\begin{tabular}{c|c}
		
		Axiom schema & Logic \\ 
		\hline 
		\begin{minipage}{30ex}
			\begin{description}
				\item[D] $\Box A \to \Diamond A$
				\item[T] $\Box A \to A$
				\item[4] $\Box A \to \Box\Box A$
			\end{description}
		\end{minipage}
		& 
		\begin{minipage}{30ex}
			\begin{tabular}{l c l}
				\KK & $=$ & $\logm{\emptyset}$\\
				\D & $=$ & $\logm{\textbf{D}}$\\
				\T & $=$ & $\logm{\textbf{T}}$\\
				\K4 & $=$ & $\logm{\textbf{4}}$\\
		 \rosso{\DQ} & $=$ & \rosso{$\logm{\textbf{D, 4}}$}\\
				\S4 & $=$ & $\logm{\textbf{T, 4}}$
			\end{tabular}
		\end{minipage} 
	\end{tabular} 
	\caption{Axioms for systems ${\KK}$,  ${\D}$,  ${\T}$,  ${\K4}$,  ${\S4}$}\label{fig:modlogics}
\end{figure}

We will call ${\D}$, ${\T}$, \rosso{$\DQ$}, and  ${\S4}$ \textit{total modal logics}, since in their Kripke semantics the accessibility relation is total. Instead, we will call ${\KK}$ and  ${\K4}$ \textit{partial modal logics}.

\section{Natural Deduction Systems}\label{sec:natdedsyst}
 In this section we define natural deduction systems for the class of logics we previously introduced.

\subsection{{Total logics}}
We start by defining the system $\nS4$. The set of derivations from a set $\Gamma$ of assumptions is defined
as the least set that contains $\Gamma$ and is
closed under application of the following rules (where, as usual, a formula into square
brackets represents a discharged assumption):

\subsubsection*{Logical rules}\label{Sect-Total-Rules}
\begin{center}
   $\brule { \LT{\pf{A}{\alpha}}{  }} { \LT{\pf{B}{\alpha}}{ } } { \pf{A\land 
B}{\alpha} }
     {\ (\land I)} $
\qquad $\urule{\LT{\pf{A \land B}{\alpha}}{ }}{\pf{ A}{\alpha}}{\ (\land_1 E)}$ 
\qquad $\urule{\LT{\pf{A \land B}{\alpha}}{ }}{\pf{ B}{\alpha}}{\ (\land_2 E)}$
\end{center}

\bigskip

\begin{center}
   $\urule{\LT{\pf{A}{\alpha}}{ }}{\pf{A \lor B}{\alpha}}{\ (\lor_1 I)}$ \quad
   $\urule{\LT{\pf{B}{\alpha}}{ }}{\pf{A \lor B}{\alpha}}{\ (\lor_2 I)}$ \quad
   $\trule{\LT{\pf{A \lor B}{\alpha}}{ }} {\LT{\pf{C}{\beta}}{[\pf{A}{\alpha}]\  }}
   {\LT{\pf{C}{\beta}}{[\pf{B}{\alpha}]\ }} {\pf{C}{\beta}} {\ (\lor E)} $
\end{center}

\bigskip

\begin{center}
   $\urule { \LT{\pf{B}{\alpha}}{[\pf{A}{\alpha}]\  } } { \pf{A\to B}{\alpha} } {\ (\to 
I)} $ \qquad $
   \brule{\LT{\pf{A\to B}{\alpha}}{ }}{\LT{\pf{A}{\alpha}}{  }}{\pf{B}{\alpha}}{\ (\to 
E)} $
\end{center}

\bigskip

\begin{center}
   $\urule { \LT{\pf{\bot}{\beta}}{[\pf{\lnot A}{\alpha}]\  } } { \pf{A}{\alpha} } {\ 
(\bot_c)} $ \qquad $
   \urule { \LT{\pf{\bot}{\beta}  }{ } } { \pf{A}{\alpha} } {\ (\bot_i)}$
\end{center}

 In $\bot_i,$  $A$ is atomic; moreover, when $A$
is $\bot$ we  require $\alpha\neq \beta.$

\bigskip

\begin{center}
   $ \urule { \LT{\pf{A}{\alpha\conc x}}{ } } { \pf{\Box A}{\alpha} } {\
     (\Box I)* } $ \qquad $ %
   \urule { \LT{\pf{\Box A}{\alpha}}{ } } { \pf{A}{\alpha\conc \beta} } { (\Box E)
     } $
\end{center}

In the rule $\Box I$, one has $\alpha\conc x \not\in \iniz{\Gamma}$, where $\Gamma$ is the
set of {(open)} assumptions on which $\pf{A}{\alpha\conc x}$ depends.

\bigskip

\begin{center}
   $ \urule { \LT{\pf{A}{\alpha\conc \beta}}{ } } { \pf{\Diamond A}{\alpha} } { \
     (\Diamond I) } $ \qquad $\brule{\LT{\pf{\Diamond A}{\alpha}}{ }}
   {\LT{\pf{C}{\beta}}{[\pf{A}{\alpha\conc x}]\ }} {\pf{C}{\beta}} {\qquad (\Diamond 
E)*} $
\end{center}

In the rule $\Diamond E$,  one has $\alpha\conc x \not\in \iniz{\beta}$ and $\alpha\conc x \not\in \iniz{\Gamma}$, where $\Gamma$ is the
set of {(open)} assumptions on which  $\pf{C}{\beta}$ depends, with the
exception of the discharged assumptions $\pf{A}{\alpha\conc x}$.

It is easy to show the admissibility of the following rule, where the requirement of atomicity of the conclusion is removed: 
\begin{center}
   $\urule { \LT{\pf{\bot}{\beta}  }{ } } { \pf{A}{\alpha} } {\ (\bot_i\mbox{-ext})}$
\end{center}
for $\pf{A}{\alpha}\neq\pf{\bot}{\beta}$.

\bigskip

On the basis of $\nS4$,  the natural deduction systems for the logics $\D$, $\T$, and $\DQ$ can be obtained 
by imposing suitable constraints on the application of  $\Box E$ and $\Diamond I$ rules, as shown in the following table.
\bigskip
\\
\begin{tabular}{|c|c|}
	\hline 
	\textbf{name of the calculus} & \textbf{constraints on  the rules $\Box E$ and $\Diamond I$}\\ 
	\hline 
	$\nS4$ & no constraints \\ 
	\hline 
	$\nT$  & $\beta=<\ >$ \\ 
	\hline 
	$\nD$  & $\beta$ is a singleton sequence $<z>$  \\ 
	\hline 
  \rosso{$\nDQ$}  & \rosso{$\beta$ is non empty} \\ 
	\hline 
\end{tabular} 
\bigskip

Let $\NN$ be one of $\nT$, $\nD$, \rosso{$\nDQ$}, $\nS4$;  
as usual we  write $\Gamma\vdash_{\NN} \pf{A}{\alpha}$ if there is a deduction $\Pi$ in $\NN$ with conclusion $\pf{A}{\alpha}$, whose non discharged assumptions appear in $\Gamma$. 
\rosso{
\begin{definition}[Proper position]
We refer to the position $\alpha x$ that explicitly appears in any of the
rules $\Box I$, $\Diamond E$ as to the \textit{proper position} of the
corresponding rule. We say that a position \emph{is  proper  in a
derivation} if it is the proper position of some $\Box I$, $\Diamond E$
rule in the derivation.
\end{definition}
}

By position renaming we can we can prove the following (see~\cite[Vol. 2, pag.
529]{tvd1988} for the analogous proof for proper \emph{variables})\footnote{To be pedantic: a position occurs in a derivation if it occurs as a prefix of $\alpha$ for some position-formula $\pf{A}{\alpha}$ of the derivation.}:

\begin{proposition}\label{vc-cond}%
   Let $\Gamma\vdash_{\NN} \pf{A}{\alpha}$.  Then
   there exists a deduction of $\pf{A}{\alpha}$ from $\Gamma$ in the system
   $\NN$ such that
   \begin{enumerate}
   \item each proper position is the proper position of exactly one instance
     of $\Box I$ or $\Diamond E$ rule;
   \item the proper position of any instance of $\Box I$ rule occurs only
     in the sub-derivation above that instance of the rule;
   \item the proper position of any instance of $\Diamond E$
     rule occurs only in the sub-derivation above the minor premiss of
     that instance of the rule.
   \end{enumerate}
\end{proposition}

\begin{definition}[Position condition] A deduction satisfying conditions 1--3  of 
Proposition~\ref{vc-cond} is
said to satisfy  the \textit{position condition}.
\end{definition}

By  Proposition~\ref{vc-cond} we can always assume that all deductions
satisfy the position condition.
We denote
by $\Pi\rep{\beta}{\gamma}$ the tree obtained by replacing each position $\alpha$ in
a deduction $\Pi$ with $\alpha\rep{\beta}{\gamma}$.
\rosso{
\begin{remark}\label{subst} Under reasonable assumptions, this operation of position substitution $\Pi\rep{\beta}{\gamma}$ preserves the position condition. Indeed, if:
\begin{enumerate}
\item $\Pi$ is a deduction satisfying the position condition;
\item  $\beta$ is a position that is
not a proper position of $\Pi;$
\item $\gamma$ is a position not
containing any proper position of $\Pi;$
\end{enumerate}
then $\Pi\rep{\beta}{\gamma}$ is a deduction satisfying  the position condition.
\end{remark}
}

Note that if the last rule of 
$\Pi$ is $\bot_i$, and the last formula is $\pf{\bot}{\alpha}$ for some 
$\alpha$,
it might be the case that, after the position substitution, the side condition of this application of $\bot_i$ is no longer satisfied (that is, its premise and conclusion are both $\pf{\bot}{\delta}$, for the same $\delta$). In such a case by  $\Pi\rep{\beta}{\gamma}$ we 
mean the deduction obtained by deleting,  after the substitution, the last---incorrect---application of 
$\bot_i.$

Finally, we want to make sense of the operation $\Pi\rep{\beta}{\gamma}$ even when the 
conditions of Remark~\ref{subst} are not satisfied. Notice that if $\Pi$ 
is a deduction satisfying the position condition, we can replace any proper position in $\Pi$ 
by a new position, to obtain a deduction $\Pi'$ of the same formula 
from the same assumptions, and such that $\beta$ and $\gamma$ satisfy all the 
conditions of Remark~\ref{subst}. Hence we define $\Pi\rep{\beta}{\gamma}$ as this
$\Pi'\rep{\beta}{\gamma}$. In the sequel we will implicitly assume that by $\Pi\rep{\beta}{\gamma}$ 
we actually mean $\Pi'\rep{\beta}{\gamma}$, for some $\Pi'$ as above.

\subsection{Weak Completeness}
We prove a Weak Completeness theorem passing through some auxiliary results. 

\begin{proposition}\mbox{}\label{Prop:tot:weakcompl}
	\begin{enumerate}
		
		\item Let $\NN$ be one of the  systems ${\nD}$, ${\nT}$, \rosso{$\nDQ$}, ${\nS4}$: $\vdash_{\NN} 
			\pf{\Diamond A \leftrightarrow \neg\Box\neg A}{<>}$;
		\item Let $\NN$ be one of the  systems ${\nD}$, ${\nT}$, \rosso{$\nDQ$}, ${\nS4}$: $\vdash_{\NN} \pf{\Box (A\to B)\to (\Box A \to \Box B)}{<>}$;
		\item \label{Prop:tot:weakcompl:T} Let $\NN$ be one of the  systems ${\nT}$,  ${\nS4}$:
		$\vdash_{\NN} \pf{ \Box A \to A}{<>}$;
		\item \label{Prop:tot:weakcompl:D} Let $\NN$ be one of the  systems ${\nD}$, \rosso{$\nDQ$}, ${\nS4}$:
		$\vdash_{\NN} \pf{ \Box A \to \Diamond A}{<>}$;
		\item \label{Prop:tot:weakcompl:4} Let $\NN$ be one of the  systems \rosso{$\nDQ$}, ${\nS4}$: 
		$\vdash_{\NN} \pf{ \Box A \to \Box\Box A}{<>}$;
	\end{enumerate}

\end{proposition}
\begin{proof}\mbox{}
\begin{enumerate}
	\item \mbox{\\}
	\def\uno{
	\prova{[\pf{\Box \neg A}{< >}]}{\pf{\neg A}{x}}{\Box E}
}
\def\due{
	\prova{
	[\pf{A}{x}] \uno
	}
	{\pf{\bot}{x}}
	{\to E}
}

\def\tre{
	\prova{
	\due
	}
	{\pf{\neg\Box\neg A}{<>}}
	{\to I}
}

\def\quattro{
	\prova{
	[\pf{\Diamond A}{<>}]	\tre
	}
	{\pf{\neg\Box\neg A}{<>}}
	{\Diamond E}
}
\def\cinque{
	\prova{
	\quattro
	}
	{\pf{\Diamond A \to \neg\Box\neg A}{<>}}
	{\to I}
}
\mbox{} \qquad $\cinque$	
\def\uno{
	\prova{[\pf{A}{x}]}{\pf{\Diamond A}{<>}}{\Diamond I}
}	
\def\due{
	\prova{[\pf{\neg\Diamond A}{x}] \uno }{\pf{\bot}{<>}}{\to E}
}
\def\tre{
	\prova{\due }{\pf{\neg A}{x}}{\to I}
}
\def\quattro{
	\prova{\tre }{\pf{\Box \neg A}{}}{\Box I}
}
\def\cinque{
	\prova{\quattro [\pf{\neg\Box\neg A}{<>} ]}{\pf{\bot}{<>}}{\to E}
}
\def\sei{
	\prova{\cinque }{\pf{\Diamond A}{<>}}{\bot_c}
}
\def\sette{
	\prova{\sei }{\pf{\neg\Box\neg A\to \Diamond A}{<>}}{\to I}
}	
	$\sette$
	
	\item \def\auno{\prova{[\pf{\Box A}{< >}]}{\pf{ A}{x}}{\Box E}}
	\def\adue{\prova{[\pf{\Box (A\to B)]}{< >}}{\pf{A\to B}{x}}{\Box E}}
	\def\atre{\prova{\auno\ \adue}{\pf{B}{x}}{\to E }}
	\def\aquattro{\prova{\atre}{\pf{\Box B}{<>}}{\Box I}}
	\def\acinque{\prova{\aquattro}{\pf{\Box A \to \Box B}{<>}}{\to I}}
	\def\asei{\prova{\acinque}{\pf{\Box (A\to B)\to (\Box A \to \Box B)}{<>}}{\to I}}
	$$\asei$$
	\item 
	\def\uno{\prova{[\pf{\Box A}{< >}]}{\pf{ A}{<>}}{\Box E}}
	\def\due{\prova{\uno}{\pf{ \Box A \to A}{<>}}{\to I}}
	$$\due$$
	\item 
	\def\uno{\prova{[\pf{\Box A}{< >}]}{\pf{ A}{x}}{\Box E}}
	\def\due{\prova{\uno}{\pf{ \Diamond A}{<>}}{\Diamond I}}
	\def\tre{\prova{\due}{\pf{ \Box A \to \Diamond A}{<>}}{\to I}}
	$$\tre$$
	\item 
	\def\uno{\prova{[\pf{\Box A}{< >}]}{\pf{ A}{xy}}{\Box E}}
	\def\due{\prova{\uno}{\pf{ \Box A}{x}}{\Box I}}
	\def\tre{\prova{\due}{\pf{ \Box \Box A}{<>}}{\Box I}}
	\def\quattro{\prova{\tre}{\pf{ \Box A \to \Box \Box A}{<>}}{\to I}}
	$$\quattro$$
\end{enumerate}

\end{proof}

Closure under \textbf{NEC}  is obtained by showing that all positions in a provable sequent may be ``lifted'' by any prefix. Observe first that, for $\Gamma=A_1^{\gamma_1},\ldots, A_n^{\gamma_n}$,  we have $\Gamma\rep{<>}{\beta} = A_1^{\beta\conc\gamma_1},\ldots, A_n^{\beta\conc\gamma_n}$.

\begin{proposition}[lift]\label{lift:basics}
	Let $\NN$ be one of the  systems ${\nD}$, ${\nT}$, \rosso{$\nDQ$}, ${\nS4}$, and let $\beta$ be a  position.
	If 
	$\Gamma\vdash_{\NN} \pf{A}{\alpha}$,
	 then 
	$\Gamma\rep{<>}{\beta}\vdash_{\NN}  \pf{A}{\beta\conc\alpha}$.
\end{proposition}
\begin{proof}
	Standard induction on derivation (with suitable renaming of proper positions). It is easily verified that the constraints on the modal rules remain satisfied.
\end{proof}

\begin{corollary}
	Let $\NN$ be one of the  systems ${\nD}$, ${\nT}$, \rosso{$\nDQ$}, ${\nS4}$.
	\\
	If $\vdash_{\NN} \pf{A}{<\ >}$, then
	$\vdash_{\NN} \pf{\Box A}{<\ >}$.
\end{corollary}

Finally, closure under \textbf{MP} is trivially ensured by rule $(\to E)$.

\begin{theorem}[weak completeness]
	Let $\sys$ be one of the modal systems {\D}, {\T}, \rosso{\DQ}, {\S4}. If  
	$\vdash_\sys A$, then $\vdash_{\ns{\sys}} \pf{A}{<>}$.
\end{theorem}

\section{{Partial logics}}\label{sec:partlogics}
The treatment of partial logics $\KK$ and $\K4$ is delicate and requires the introduction of auxiliary notions to soundly define their formal system and prove proof-theoretic results.
To motivate the formal systems for $\KK$ and $\K4$, remember that in the semantics of Section~\ref{subs:seman}, positions will be mapped into nodes of a Kripke structure. 
Both $\KK$ and $\K4$ are complete with respect to the class of models where the accessibility relation is not always defined. This means that \emph{the correspondence between positions and nodes could be undefined at some position}, a situation reminiscent of the case of first order logic with undefined terms\footnote{The formal analogy between variables/terms and tokens/positions (and hence between quantifiers and modalities) is one of the \emph{leitmotive} of the 2-sequents approach, as we already mentioned in Section~\ref{sec:2seq}).}. In fact, we will treat this case with an \textit{existence predicate for positions}, a tool introduced by D. Scott in the late seventies~\cite{Scott1979} to deal with empty domains, and therefore with partially defined terms.
For a first order logic term $t$, the predicate $\ext{t}$ has the following intuitive meaning:  \textit{$t$ is defined}\footnote{For an extensive treatment of existence predicates for first order natural deduction, see the two volumes~\cite{tvd1988}, or the survey~\cite{BaazIemhoff2005}.}.

The natural deduction systems introduced in the previous section are now expanded with formulas of the form $\ext{\alpha}$, where $\alpha$ is any position and which we informally read as:  \textit{$\alpha$ denotes an existing node/object.} Such formulas $\ext{\alpha}$ \emph{may be used only as premises} in deductions. The only modified rules w.r.t. the previously introduced formal  system are the modal ones. 

Rules for $\Box$ are the following:
\begin{center}
	$ \urule { \LT{\pf{A}{\alpha\conc x}}{[\ext{\alpha\conc x}] } } { \pf{\Box A}{\alpha} } {\
		(\Box I)* } $ \qquad $ %
	\urule { \LT{\pf{\Box A}{\alpha}}{ }\ext{\alpha\conc \beta}} { \pf{A}{\alpha\conc \beta} } { (\Box E)
	} $
\end{center}
where in the rule $\Box I$, $\alpha\conc x \not\in \iniz{\Gamma}$, where $\Gamma$ is the
set of {(open)} assumptions on which $\pf{A}{\alpha\conc x}$ depends.

Rules for $\Diamond$ are the following:

\bigskip

\begin{center}
	$ \urule { \LT{\pf{A}{\alpha\conc \beta}  }\qquad \ext{\alpha\conc \beta}{ } } { \pf{\Diamond A}{\alpha} } { \
		(\Diamond I) } $ \qquad $\brule{\LT{\pf{\Diamond A}{\alpha}}{ }}
	{\LT{\pf{C}{\beta}}{[\pf{A}{\alpha\conc x}]\quad [\ext{\alpha\conc x}]\ }} {\pf{C}{\beta}} {\qquad (\Diamond 
		E)*} $
\end{center}
where in rule $\Diamond E$,  $\alpha\conc x \not\in \iniz{\beta}$ and $\alpha\conc x \not\in \iniz{\Gamma}$, where $\Gamma$ is the
set of {(open)} assumptions on which  $\pf{C}{\beta}$ depends, with the
exception of the discharged assumptions $\pf{A}{\alpha\conc x}$.

These ``generic'' rules are further constrained to take into account the specifics of the systems $\K4$ and $\KK$. The following table gives such constraints for the systems $\textbf{N}_{\textbf{K4}}$ and $\textbf{N}_{\textbf{K}}$.

\bigskip
\begin{tabular}{|c|c|}
	\hline 
	\textbf{name of the calculus} & \textbf{constraints on  the rules $\Box E$ and $\Diamond I$}\\ 
	\hline 
	$\nK4$ & $\beta$ is a non empty sequence \\ 
	\hline 
	$\nKK$  & $\beta$ is a singleton sequence $<z>$  \\ 
	\hline 
\end{tabular} 

\subsection{Weak Completeness}

We  prove a Weak Completeness Theorem also for partial logics.

\begin{proposition}\mbox{}\label{prop:partial:weakcomp}
	\begin{enumerate}
			\item Let $\NN$ be one of the  systems ${\nKK}, {\nK4}$, $\vdash_{\NN} 
		\pf{\Diamond A \leftrightarrow \neg\Box\neg A}{<>}$;
		\item \label{prop:partial:weakcomp:K}
		Let $\NN$ be one of the  systems ${\nKK},  {\nK4}$, $\vdash_{\NN} \pf{\Box (A\to B)\to (\Box A \to \Box B)}{<>}$;
		\item $\vdash_{\nK4} \pf{ \Box A \to \Box\Box A}{<>}$;
	\end{enumerate}
	
\end{proposition}
\begin{proof}{In the following derivations, observe the interplay between modal introduction and elimination, which allows to discharge all existence predicates.}
	\begin{enumerate}
			\item \mbox{\\}
		\def\uno{
			\prova{[\pf{\Box \neg A}{< >}] [\ext{x}]}{\pf{\neg A}{x}}{\Box E}
		}
		\def\due{
			\prova{
				[\pf{A}{x}] \uno
			}
			{\pf{\bot}{x}}
			{\to E}
		}
		
		\def\tre{
			\prova{
				\due
			}
			{\pf{\neg\Box\neg A}{<>}}
			{\to I}
		}
		
		\def\quattro{
			\prova{
				[\pf{\Diamond A}{<>}]	\tre
			}
			{\pf{\neg\Box\neg A}{<>}}
			{\Diamond E}
		}
		\def\cinque{
			\prova{
				\quattro
			}
			{\pf{\Diamond A \to \neg\Box\neg A}{<>}}
			{\to I}
		}
		\mbox{} \qquad $\cinque$	
		\def\uno{
			\prova{[\pf{A}{x}] [\ext{x}]}{\pf{\Diamond A}{<>}}{\Diamond I}
		}	
		\def\due{
			\prova{[\pf{\neg\Diamond A}{x}] \uno }{\pf{\bot}{<>}}{\to E}
		}
		\def\tre{
			\prova{\due }{\pf{\neg A}{x}}{\to I}
		}
		\def\quattro{
			\prova{\tre }{\pf{\Box \neg A}{}}{\Box I}
		}
		\def\cinque{
			\prova{\quattro [\pf{\neg\Box\neg A}{<>} ]}{\pf{\bot}{<>}}{\to E}
		}
		\def\sei{
			\prova{\cinque }{\pf{\Diamond A}{<>}}{\bot_c}
		}
		\def\sette{
			\prova{\sei }{\pf{\neg\Box\neg A\to \Diamond A}{<>}}{\to I}
		}	
		$\sette$

		\item \def\auno{\prova{[\pf{\Box A}{< >}] [\ext{x}]}{\pf{ A}{x}}{\Box E}}
		\def\adue{\prova{[\pf{\Box (A\to B)]}{< >} [\ext{x}] }{\pf{A\to B}{x}}{\Box E}}
		\def\atre{\prova{\auno\ \adue}{\pf{B}{x}}{\to E }}
		\def\aquattro{\prova{\atre}{\pf{\Box B}{<>}}{\Box I}}
		\def\acinque{\prova{\aquattro}{\pf{\Box A \to \Box B}{<>}}{\to I}}
		\def\asei{\prova{\acinque}{\pf{\Box (A\to B)\to (\Box A \to \Box B)}{<>}}{\to I}}
		$$\asei$$
		\item 
		\def\uno{\prova{[\pf{\Box A}{< >}][\ext{xy}]}{\pf{ A}{xy}}{\Box E}}
		\def\due{\prova{\uno}{\pf{ \Box A}{x}}{\Box I}}
		\def\tre{\prova{\due}{\pf{ \Box \Box A}{<>}}{\Box I}}
		\def\quattro{\prova{\tre}{\pf{ \Box A \to \Box \Box A}{<>}}{\to I}}
		$$\quattro$$
	\end{enumerate}
	
\end{proof}

Closure under \textbf{NEC} and under \textbf{MP} is shown in the same manner as for the total systems. Therefore:
	
	\begin{theorem}[weak completeness]
		Let $\sys$ be one of the modal systems {\KK} and {\K4}. If  
		$\vdash_\sys A$ then $\vdash_{\ns{\sys}} \pf{A}{<>}$.
	\end{theorem}

% !TEX root = ./NatDedPos-revised.tex
%%%%%%%%%%%%%%%%%%%%%%%%%%%%%%%%
%%%%%%%%%semantics%%%%%%%%%
%%%%%%%%%%%%%%%%%%%%%%%%%%%%%%%%
\section{Semantics}\label{subs:seman}

We introduce in this section a tree-based  Kripke semantics for our modal systems, to prove their completeness with respect to the standard axiomatic presentations. 

\subsection{Trees and Tree-semantics}
Let $\NN^*$ be the set of finite sequences of natural numbers with the partial order $\less{\NN}$ as defined in Section~\ref{sect:sequences}.

\begin{definition}
A \textit{tree}  is  a subset  $\tree$ of $ \NN^*$  s.t. 
$<\ >\in \tree$; and
if $t\in \tree$ and $s \less{\tree} t$, then $s\in \tree, $ 
where $\less{\tree}$ is the restriction of $\less{\NN}$ to $\tree$.
\end{definition}
\noindent The elements of $\tree$ are called
\textit{nodes}; a \textit{leaf} is a node with no successors.
Given a tree $\tree$ and $s\in \tree$, we define $\tree_s$ (the \emph{subtree of $\tree$ rooted at $s$}) to be the tree  defined
as:
$ s'\in \tree_s\ \Leftrightarrow\ s\conc s' \in \tree$.
Observe that 
$\tree_{<\ >} = \tree$.
In this section, $s$ and $t$ will range over the generic elements (nodes) of $\Theta$.

If $At$ is the set of proposition symbols of our modal language, a \emph{Kripke model} is a triple $\mdl{}=<\tree,\nu, \R>$, where $\tree$ is a tree, $\nu:\tree\to 2^{At}$ is an assignment of proposition symbols to nodes, and  $\R\subseteq \tree\times\tree$.
Given a modal system $\sys\in{\{{\KK},{\D},{\T},{\K4},{\S4}\}}$, a \emph{$\sys$-model} is a Kripke model  $\mdl{_\sys}=<\tree,\nu, \R>$ s.t.
\bigskip

\begin{tabular}{|c|c|c|}
 \hline 
 \textbf{modal system} & \textbf{conditions on $\tree$} & \textbf{conditions on $\R$} \\ 
 \hline 
 {\KK}&  no condition  & $\R=\suc{\tree}$ \\ 
 \hline 
 {\D} & $\tree$ does not have leaves &  $\R=\suc{\tree}$  \\ 
 \hline 
 {\T} & no condition & $\R=\suc{\tree}^0$\\ 
 \hline 
 {\K4} & no condition  & $\R=\lest{\tree}$ \\ 
 \hline 
 \rosso{\D4} & \rosso{$\tree$ does not have leaves}  & \rosso{$\R=\lest{\tree}$} \\ 
 \hline 
 {\S4} & no condition & $\R=\less{\tree}$  \\ 
 \hline 
 \end{tabular}  
 
\bigskip\noindent
The satisfiability  relation of formulas on a Kripke model is standard; e.g., for a model $\mdl{}$ and node $s$,
$\mdl{},s \models \Box A \Leftrightarrow \forall t. s R t \Rightarrow \mdl{}, t\models A$.
As usual, we write $ \mdl{} \models A$, when $\mdl{},s \models  A$ for all nodes $s$ of $\mdl{}$.

\begin{theorem}[standard completeness]
For each modal system  $\sys$ in {\KK}, {\D}, {\T}, {\K4}, \rosso{\DQ}, {\S4}, and for every formula $A$,
$\vdash_{\sys} A$ $\Leftrightarrow$  for all $\sys$-model $\mdl{}$, we have $\mdl{} \models A$.
\end{theorem}

In the following, semantics definitions and the soundness theorem are given separately for total logics (Section~\ref{semantics-sequents}) and for partial logics (Section~\ref{sec:soundpartiallog}).

\subsection{Semantics: Total logics}\label{semantics-sequents}
 \newcommand{\conv}[1]{{#1}\!\downarrow}

\begin{definition}[Structures]
Let $\sys\in{\{{\D}, {\T}, \rosso{\DQ}, {\S4}\}}$ be a modal system. 
A $2_\sys$ \emph{structure} is a pair $\tmdl{\tree}=<\mdl{\tree},\rho>$ where:
\begin{itemize}
	\item $\mdl{\tree}$ is an $\sys$-model $<\tree,\nu, R>$
	\item $\rho:\pos \to \tree$ is a map from positions to nodes (the \emph{evaluation}).
	\item[] Moreover for  $\alpha\in \pos$, {and for a fixed $\rho$,} with $\srho{\alpha}$ we denote an evaluation $\srho{\alpha}:\pos \to \Theta_{\rho(\alpha)}$.
\end{itemize}
\end{definition}

Depending on the specific modal system, $\rho$ has to satisfy the following, additional constraints:

\bigskip
\begin{tabular}{|c|c|}
 \hline 
 \textbf{modal system} & \textbf{conditions on $\rho$} \\ 
 \hline 
 \D &  $ \rho \mbox{\ is total\ }\ \&\ (\alpha\suc{\pos} \beta \Rightarrow  \rho(\alpha)\suc{\tree\ } \rho(\beta))$  \\ 
 \hline 
 \T & $\rho \mbox{\ is total\ }\ \&\ (\alpha\suc{\pos} \beta \Rightarrow  \rho(\alpha)\suc{\tree\ }^0 \rho(\beta))$\\ 
 \hline 
 \rosso{\DQ} & \rosso{ $ \rho \mbox{\ is total\ }\ \&\ (\alpha\suc{\pos} \beta \Rightarrow  \rho(\alpha)\lest{\tree\ } \rho(\beta))$}  \\ 
 \hline 
 \S4 &  $ \rho \mbox{\ is total\ }\ \&\ ( \alpha\suc{\pos} \beta \Rightarrow  \rho(\alpha)\less{\tree}\rho(\beta))$  \\ 
 \hline 
 \end{tabular}  
\bigskip

The satisfiability relation $\Vdash$ between a 2-structure and a position formula is defined in the following way:

\[
\mdl{\tree},\rho\Vdash  \pf{A}{\alpha} \Leftrightarrow  	 \mdl{\tree},\rho(\alpha) \models A,\]
where $\models$ is the standard satisfiability relation w.r.t. modal Kripke semantics.

Finally, given a modal system $\sys$, we define the notion of logical consequence for positions formulas.
Let $\sys$ be one of the systems $\T$, $\D$, $\rosso{\DQ}$, or $\S4$: 
$$
\Gamma\Vdash_{\sys}\pf{A}{\alpha}
\Leftrightarrow
\forall \langle \mdl{\tree},\rho\rangle .
(
\forall \pf{B}{\alpha}\in\Gamma \; \mdl{\tree},\rho\Vdash \pf{B}{\alpha})
\Rightarrow 
\mdl{\tree},\rho\Vdash \pf{A}{\alpha}.
$$

We now introduce some notation for the semantical substitution of values into the evaluation function $\rho$, in correspondence of specific subtrees.
For $t\in\Theta$ and $\srho{\alpha}$, define 
\[
\rho\{\alpha\conc x/\srho{\alpha}\}(\beta)=
\begin{cases}
\rho(\beta) \mbox{\qquad \qquad\ \ if\ } \beta\neq \alpha\conc x\conc \gamma \\
\rho(\alpha)\conc\srho{\alpha}(x\conc \gamma ) \mbox{\quad otherwise\ }\\
\end{cases}
\]

{
We define the following set of $\Theta$ elements:
\begin{itemize}
	\item  $\Theta_{\D}=\{t: |t|= 1\} $;
	\item  $ \Theta_{\T}=\{t: |t|\leq 1\}$;
	\item  $ \Theta_{\S4}=\{t: |t|\geq 0\}$.
\end{itemize}
}

As for other notations, we will write $\Theta_\sys$ for any of these sets. 

Let us now fix a specific structure $<\mdl{\tree},\rho>$;  we have the following. 

\begin{lemma}\label{lemma:sub1} Let $\sys\in\{ {\D}, {\T},\rosso{\DQ}, {\S4} \}$. 
	\begin{enumerate}
		\item 
			$\mdl{\tree}, \rho\Vdash \Box \pf{A}{\alpha} 
				\Leftrightarrow 
					 \forall \srho{\alpha}. 
					\mdl{\tree},\rho\{\alpha\conc x/\srho{\alpha}\}\models  \pf{A}{\alpha\conc x};$
\item $\mdl{\tree}, \rho\Vdash \Diamond \pf{A}{\alpha} 
	\Leftrightarrow 
		 \exists \srho{\alpha}. \mdl{\tree}, 
		 \rho\{\alpha\conc x/\srho{\alpha}\} \models \pf{A}{\alpha\conc x}.$
	\end{enumerate}
\end{lemma}
\begin{proof}
\mbox{}\\
$
\mdl{\tree},\rho\Vdash \Box \pf{A}{\alpha} \\
\mbox{}\qquad \Leftrightarrow\\
\mdl{\tree},\rho(\alpha)\models \Box A \\
\mbox{}\qquad 
\Leftrightarrow \\
\forall t, \mdl{\tree},\rho(\alpha)\conc t\in\Theta \Rightarrow \rho(\alpha)\conc t \models A 
\\
\mbox{}\qquad 
\Leftrightarrow \\
\forall  \srho{\alpha}. 
\mdl{\tree},\rho\{\alpha\conc x/\srho{\alpha}\}(\alpha\conc x) \models A 
\\
\mbox{}\qquad 
\Leftrightarrow \\
\mdl{\tree},\rho\{\alpha\conc x/\srho{\alpha}\}\Vdash  \pf{A}{\alpha\conc x}.
$

\end{proof}

Let $v \R u$ in a tree $\Theta$, we define the subtraction operation $\div$ between nodes as:
$v\div u = t \Leftrightarrow u\conc t = v$

\begin{lemma}\label{lemma:diff1}
Let $\srho{\alpha}$ be an evaluation s.t.
$\srho{\alpha}(x) = \rho(\alpha\conc\beta) \div \rho(\alpha)$, then
\[
\mdl{\tree},\rho\models \pf{A}{\alpha\conc\beta} \Leftrightarrow 
\mdl{\tree},\rho\{\alpha\conc x/\srho{\alpha}\}\models  \pf{A}{\alpha\conc x}.
\]
\end{lemma}
\begin{proof} Observe that 
$
\rho\{\alpha\conc x/\srho{\alpha}\}
=
\rho(\alpha)\conc \srho{\alpha}(x)
=
\rho(\alpha)\conc (\rho(\alpha\conc\beta) \div \rho(\alpha))=\rho(\alpha\conc\beta)
$; 
therefore
\\
$
\mdl{\tree},\rho\Vdash  \pf{A}{\alpha\conc\beta} \\
\mbox{}\qquad \Leftrightarrow \\
\mdl{\tree},\rho(\alpha\conc\beta) \models A\\
\mbox{}\qquad \Leftrightarrow \\
\mdl{\tree},\rho\{\alpha\conc x/\srho{\Theta_\alpha}\}(\alpha\conc x)\models A\\
\mbox{}\qquad \Leftrightarrow \\
\mdl{\tree},\rho\{\alpha\conc x/\srho{\Theta_\alpha}\}\Vdash  \pf{A}{\alpha\conc x}.
$
\end{proof}

We are finally in the position to prove the \emph{soundness} theorem, by an easy induction on proofs which---we remark once again---strictly mimics the standard proof of soundness for  first order natural deduction. 
In the rest of the paper with $hp(\Pi)$ we denote the set of undischarged hypoteses of the deduction $\Pi$.
We write 
  $$\ded{\Pi}{R}{\pf{A}{\alpha}}$$ 
for $\Pi$ is a deduction
 	of formula $\pf{A}{\alpha}$ whose last rule is $R$.

\begin{theorem}[soundness 1]\label{thm:soundness1}
Let $\sys\in\{{\D}, {\T}, \rosso{\DQ}, {\S4}  \}$ be a modal system.\\
If $\Gamma\vdash_{\ns{\sys}} \pf{A}{\alpha}$  then $\Gamma\Vdash_{{\sys}} \pf{A}{\alpha}$.
\end{theorem}
\begin{proof}[Proof sketch.] 
	Let $\sys\in\{{\D}, {\T}, \rosso{\DQ}, {\S4}  \}$ and assume that in $\ns{\sys}$
	$$\ded{\Pi}{R}{\pf{A}{\alpha}}$$
We prove by induction on the length of $\Pi$, for each $\Gamma$ such that $hp(\Pi)\subseteq \Gamma $, that
$\Gamma\Vdash_{{\sys}} \pf{A}{\alpha}$.
We  discuss only the cases where $R$ is  $\Box I$ or $\Box E$.
\begin{description}
	\item[($\Box I$)]
	Let $\Pi$ be 
	$$\prova{
	\ded{\Pi'}{}{\pf{A}{\alpha\conc x}}
			}{\pf{\Box A}{\alpha}}{}$$

We observe first that the rule is the same for all the systems under consideration, and that 
$\alpha\conc x \not\in \iniz{hp(\Pi')}$, with $hp(\Pi')\subseteq \Gamma$.

By IH we have: 
	$
	\forall \mdl{\tree},\rho. \mdl{\tree},\rho \Vdash hp(\Pi') \Rightarrow \mdl{\tree},\rho \Vdash\pf{A}{\alpha\conc x}
	$
	 \\	
	 \mbox{}\qquad $\Leftrightarrow$ (by the genericity of $\rho$)
	 \\ 
	 $\mdl{\tree},\rho, \srho{\alpha}.  \mdl{\tree},\rho\{\alpha\conc x / \srho{\alpha} \} \Vdash hp(\Pi'), 
	 \Rightarrow
	 \mdl{\tree},\rho\{\alpha\conc x / \srho{\alpha} \} \Vdash \pf{A}{\alpha\conc x}$\\
	  \mbox{}\qquad $\Leftrightarrow$ 
	  (since $\mdl{\tree}, \rho\{\alpha\conc x / \srho{\alpha} \} \Vdash hp(\Pi') 
	  \Leftrightarrow
	  \mdl{\tree},\rho\Vdash hp(\Pi')
	   $)
	 \\  
	 $\forall \mdl{\tree},\rho.
	 		( \mdl{\tree},\rho \Vdash  hp(\Pi')  \Rightarrow 
	 		\forall \srho{\alpha}, 
	 \mdl{\tree},\rho\{\alpha\conc x / \srho{\alpha} \} \Vdash \pf{A}{\alpha\conc x}.\\
	 $
	  $\mbox{}\qquad \Leftrightarrow$ (by Lemma~\ref{lemma:sub1})\\
	  $\forall \mdl{\tree},\rho.
	  	(\mdl{\tree},\rho\Vdash hp(\Pi') \Rightarrow \mdl{\tree},\rho \Vdash \pf{\Box A}{\alpha}).
	 $
	  
	 %%%%%%%%%%%%%%%%%%%%%%%%%%%%%%%%
	\item[($\Box E$)]
	Let $\Pi$ be 
	$$\prova{
		\ded{\Pi'}{}{
			\pf{\Box A}{\alpha}}
	}{\pf{A}{\alpha\conc \beta}}{}$$
	
The rule have different constraints in different systems; we deal with the $\nS4$ case, the others being similar or easier. 
	
	We know that  $hp(\Pi')\subseteq \Gamma$, therefore 
	by IH \\
	$
	\forall \mdl{\tree},\rho. \mdl{\tree},\rho \Vdash hp(\Pi')
		\Rightarrow
				\mdl{\tree},\rho \Vdash\pf{\Box A}{\alpha}\\
	\mbox{}\qquad \Leftrightarrow $ (by  Lemma~\ref{lemma:sub1}) \\
	%%%
	$
		\forall \mdl{\tree},\rho, \srho{\alpha}. \mdl{\tree},
		\rho \Vdash hp(\Pi')
		\Rightarrow
	\rho\{\alpha\conc x / \srho{\alpha} \} \Vdash \pf{A}{\alpha\conc x}
	\\
	\mbox{}\qquad \Rightarrow $ (by taking $\srho{\alpha}$ s.t.  $\srho{\alpha}(x)=\rho(\alpha\conc\beta)\div\rho(\alpha)) $\\
	$
		\forall \mdl{\tree},\rho,
					\mdl{\tree},
					\rho \Vdash hp(\Pi')
					\Rightarrow
					\rho\{\alpha\conc x / \srho{\alpha} \} \Vdash \pf{A}{\alpha\conc x}
	$
	\\
	$\mbox{}\qquad \Rightarrow $ (by Lemma~\ref{lemma:diff1})\\
	$
		\forall \mdl{\tree},\rho,
	\mdl{\tree},
		\rho \Vdash hp(\Pi')
		\Rightarrow
	\rho\Vdash \pf{A}{\alpha\conc \beta}
	$
\end{description}
\end{proof}

\begin{corollary}\label{cor:interpr}
	Let $\sys\in\{{\D}, {\T}, \rosso{\DQ}, {\S4}  \}$ be a modal system.
 If $\vdash_{\ns{\sys}}\pf{A}{\alpha}$, then in the Hilbert-style presentation of $\sys$ we have  $\vdash_{\sys} A$.
\end{corollary}
 
\subsection{Semantics: Partial logics}\label{sec:soundpartiallog}
We now extend the semantical definitions and results of the previous section to the partial systems $\KK{}$ and  $\K4$. In particular, $\rho$ could be undefined on some position. Therefore, with respect to the semantics we have given in Section~\ref{semantics-sequents}:

\begin{enumerate}
\item $\rho:\pos \rightharpoonup \tree$ is a partial function;
\item $\srho{\alpha}:\pos \rightharpoonup \Theta_{\rho(\alpha)}$ is a partial function;
\item the substitution $\rho\{\alpha\conc x/\srho{\alpha}\}$ is undefined whenever it  formally contains an undefined  subexpression.
\end{enumerate}

We write $\rho(x)\!\downarrow$ and $\srho{\alpha}(x)\!\downarrow$ when the functions $\rho$ and  $\srho{\alpha}$ are defined on input $x$. We require that $\rho(\gamma)\!\downarrow\; \Rightarrow \forall \beta \less{} \gamma.\rho(\beta)\!\downarrow$, and $\srho{\alpha}(\beta)\!\downarrow\; \Rightarrow \forall \beta \less{} \gamma.\srho{\alpha}(\beta)\!\downarrow$. The constraints on evaluations for $\KK$ and $\K4$ are given in the following table.

\medskip
\begin{tabular}{|c|c|}
	\hline 
	\textbf{modal system} & \textbf{conditions on $\rho$} \\ 
	\hline 
	\KK&  $(\alpha\suc{\pos} \beta\ \&\  \conv{\rho(\alpha)}\ \&\  \conv{\rho(\beta)}) \Rightarrow  \rho(\alpha)\suc{\tree\ } \rho(\beta)$\\ 
	\hline 
	\K4 &  $(\alpha\suc{\pos} \beta\ \&\  \conv{\rho(\alpha)}\ \&\  \conv{\rho(\beta)}) \Rightarrow  \rho(\alpha)\lest{\tree } \rho(\beta)$ \\ 
	\hline 
\end{tabular}  
\medskip

Since $\rho$ is partial,  we now need two different notions of satisfiability: $\Vdash^\ell$ for assumption  formulas, and $\Vdash^r$ for conclusion formulas. Define then, for a  $2_\sys$ structure $<\mdl{\tree},\rho>$:
\begin{itemize}
\item
$\mdl{\tree},\rho\Vdash^\ell \pf{A}{\alpha} \Leftrightarrow  (\conv{\rho(\alpha)} \ \&\ 	 \mdl{\tree},\rho(\alpha) \models A$);
\item  
$\mdl{\tree},\rho\Vdash^r \pf{A}{\alpha} \Leftrightarrow  (\conv{\rho(\alpha)	}\ \Rightarrow \mdl{\tree},\rho(\alpha) \models A)$.
\end{itemize}

Semantics of the existence predicate $\ext{}$ justifies its name:
 $$
 \mdl{\tree},\rho\Vdash^l \ext{\alpha} \Leftrightarrow \conv{\rho(\alpha)}.
 $$

Note that we do not need to define $\Vdash^r$ for $\ext{}$, since it is used only in assumptions. Finally 

$
 \Gamma\Vdash\pf{A}{\alpha} 
 \Leftrightarrow 
\forall<\mdl{\tree},\rho>.
(
 \forall \pf{B}{\beta}\in\Gamma,. \mdl{\tree},\rho\Vdash^l \pf{B}{\beta}, \forall \ext{\delta} \in\Gamma. \mdl{\tree},\rho\Vdash^l \ext{\delta} 
 \\
\mbox{}\qquad\qquad\qquad
\Rightarrow  \mdl{\tree},\rho\models^r \pf{A}{\alpha}
 ).
$

Finally define:
\begin{itemize}
	\item $\Theta_{\KK}= \{t: |t|=1\}$;
	\item $ \Theta_{\K4}=\{t: |t|\gt0\}$.
\end{itemize}
As for other notations, we will write $\Theta_\sys$ for any of these sets. 

As for the case of total logics we have the following lemmas (the proofs are  simple adaptations of the previous ones).

\begin{lemma}\label{lemma:Partial:sub1} Let $\sys\in\{ {\KK}, {\K4}  \}$. 
	\begin{enumerate}
		\item 
		$\mdl{\tree}, \rho\Vdash^r \Box \pf{A}{\alpha} 
		\Leftrightarrow 
		\forall \srho{\alpha}. 
		\mdl{\tree},\rho\{\alpha\conc x/\srho{\alpha}\}\Vdash^t \pf{A}{\alpha\conc x};$
		\item $\mdl{\tree}, \rho\Vdash^r \Diamond \pf{A}{\alpha} 
		\Leftrightarrow 
		\exists \srho{\alpha}. \mdl{\tree}, 
		\rho\{\alpha\conc x/\srho{\alpha}\} \Vdash^r \pf{A}{\alpha\conc x}.$
	\end{enumerate}
\end{lemma}

\begin{lemma}\label{lemma:Partial:diff1}
	Let $\srho{\alpha}$ be an evaluation 
	s.t.  $\srho{\alpha}(x) = \rho(\alpha\conc\beta) \div \rho(\alpha)$, then
	\[
	\mdl{\tree},\rho\Vdash^r \pf{A}{\alpha\conc\beta} \Leftrightarrow 
	\mdl{\tree},\rho\{\alpha\conc x/\srho{\alpha}\}\Vdash^r  \pf{A}{\alpha\conc x}.
	\]
\end{lemma}

The following lemma  allows us to reuse with simple modifications the soundness theorem we proved in the previous section. 
 
\begin{lemma}\label{lemma:E-elimination}
	If $\alpha\conc x \not\in \iniz{\Gamma}$ and $\Gamma,\ext{\alpha\conc x}\Vdash \pf{\Box A}{\alpha}$ then $\Gamma \Vdash \pf{\Box A}{\alpha}$.
\end{lemma}
\begin{proof}
	
	Let us suppose that there exist $\mdl{\tree}$ and  $\rho'$ s.t.
	
	 $\mdl{\tree}, \rho'\Vdash^l \Gamma$ and $\mdl{\tree}, \rho'\not\Vdash^r\pf{\Box A}{\alpha}$.
	 
	By means of the previous lemmas we have that:
	
	$\mdl{\tree}, \rho\not\Vdash^r \Box \pf{A}{\alpha} 
	\Leftrightarrow 
	\exists \srho{\alpha}. \mdl{\tree}, 
	\rho\{\alpha\conc x/\srho{\alpha}\} \not\Vdash^r \pf{A}{\alpha\conc x}.$
	
	Now this implies that $\rho\{\alpha\conc x/\srho{\alpha}\}(\alpha\conc x)\!\downarrow$. 
	
	Let $\rho''= \rho\{\alpha\conc x/\srho{\alpha}\}$. Since $\alpha\conc x\not\in \iniz{\Gamma}$ we have  an evaluation $\rho''$ s.t. 
	
	$\mdl{\tree}, \rho''\Vdash^l \Gamma$ and $\rho''(\alpha\conc x)\!\downarrow$ and  $\mdl{\tree}, \rho''\not\Vdash^r\pf{\Box A}{\alpha}$, 
	
	which is a contradiction.
\end{proof}

\begin{theorem}[soundness 2]
	Let $\sys\in\{{\KK},   {\K4}  \}$ be a modal system.
	If $\Gamma\vdash_{\ns{\sys}} \pf{A}{\alpha}$  then $\Gamma\Vdash_{{\sys}} \pf{A}{\alpha}$.
\end{theorem}
\begin{proof}[Proof sketch.] 

	Let $\sys\in\{{\KK}, {\K4}  \}$, and assume that in $\ns{\sys}$
	$$\ded{\Pi}{R}{\pf{A}{\alpha}}$$
We prove by induction on the length of $\Pi$, for each $\Gamma$ such that $hp(\Pi)\subseteq \Gamma $, that
$\Gamma\Vdash_{{\sys}} \pf{A}{\alpha}$.
We  discuss only the cases where $R$ is  $\Box I$ or $\Box E$.
	\begin{description}
		\item[($\Box I$)] 
		Let  $\Pi$ be 
		$$\prova{
			\ded{\DT{[\ext{\alpha\conc x}]}{\Pi'}{}}{}{\pf{A}{\alpha\conc x}}
		}{\pf{\Box A}{\alpha}}{}$$
By the same argument we used in Theorem~\ref{thm:soundness1}, we have 
		$hp(\Pi') \Vdash \pf{\Box A}{\alpha}$.
		By Lemma~\ref{lemma:E-elimination} we obtain the thesis:
		$hp(\Pi')-\{\ext{\alpha\conc x}\}\Vdash \pf{\Box A}{\alpha}$.
		
		%%%%%%%%%%%%%%%%%%%%%%%%%%%%%%%%
		\item[($\Box E$)]
		Let $\Pi$ be 
		$$\prova{
			\ded{\Pi'}{}{
				\pf{\Box A}{\alpha}
			} 
		\ded{}{}{
			\ext{\alpha\conc\beta}
		}
		}
		{
			\pf{A}{\alpha\conc \beta}
		}
		{}$$
		
		We deal with the $\nK4$ case, the $\nKK$ case being similar. 
		
		We know that  $hp(\Pi')\subseteq \Gamma$, therefore 
		by IH \\
		$
		\forall \mdl{\tree},\rho. \mdl{\tree},\rho \Vdash^l hp(\Pi')
		\Rightarrow
		\mdl{\tree},\rho \Vdash^r \pf{\Box A}{\alpha}\\
		\mbox{}\qquad \Leftrightarrow $ (by  Lemma~\ref{lemma:Partial:sub1}) \\
		%%%
		$
		\forall \mdl{\tree},\rho, \srho{\alpha}. \mdl{\tree},
		\rho \Vdash hp(\Pi')
		\Rightarrow
		\mdl{\tree}, \rho\{\alpha\conc x / \srho{\alpha} \} \Vdash \pf{A}{\alpha\conc x}
		\\
		\mbox{}\qquad \Rightarrow $ (by taking $\srho{\alpha}$ s.t.  $\srho{\alpha}(x)=\rho(\alpha\conc\beta)\div\rho(\alpha)$, which 
		exists, since we assume $\ext{\alpha\conc\beta}$, that is
		$\rho({\alpha\conc\beta}) \!\downarrow$
		)
		
		$
		\forall \mdl{\tree},\rho,
		\mdl{\tree},
		\rho \Vdash^l hp(\Pi')
		\Rightarrow
		\rho\{\alpha\conc x / \srho{\alpha} \} \Vdash^r \pf{A}{\alpha\conc x}
		$
		\\
		$\mbox{}\qquad \Rightarrow $ (by Lemma~\ref{lemma:Partial:diff1})\\
		$
		\forall \mdl{\tree},\rho,
		\mdl{\tree},
		\rho \Vdash^l hp(\Pi')
		\Rightarrow
		\rho\Vdash^r \pf{A}{\alpha\conc \beta}
		$
		
	\end{description}
\end{proof}

% !TEX root = ./NatDedPos-revised.tex

\section{Intuitionistic systems and normalization}\label{sec:normalization}
\blue{We introduce intuitionistic systems, which we obtain syntactically from the ones of the previous sections in the same way intuitionistic propositional natural deduction is obtained from its classical version---by dropping the \textit{reductio ab absurdum} rule, $\bot_c$. In the economy of the paper, these intuitionistic systems are instrumental to obtain a \emph{syntactic} proof of consistency
for the classical ones\footnote{\blue{Consistency is of course already implied by the semantical results of Section~\ref{subs:seman}.}}
(Remark~\ref{tnkcons}) via a double negation translation (Section~\ref{Sect:doubleneg}). 
Contraction on proofs is defined in the standard way---on modal connectives is  defined ``as'' the one for first-order quantifiers,---and also the proof of normalization follows standard techniques. 
We see this as a further ``litmus test'' for  the simplicity and naturalness of the notion of position-formulas (and therefore this is also a second reason for the inclusion of intuitionistic calculi in the paper.) 
In future work we will explore the extraction of proof-terms (lambda-terms) from these intuitionistic systems, studying a possible Curry-Howard isomorphism for our modal systems (see Section~\ref{sec:conclusions} for more details on this.) In this paper, whose focus is on the fundamentals of the proof-theory of position-formulas, we also refrain from any attempt to discuss the formal semantics of these systems (see~\cite{Kojima2012} for a survey of some of the many  possible approaches to the semantics of intuitionistic modal logics.)  
}

Let $\iKK$, $\iT$, $\iD$, $\iK4$, \rosso{$\iDQ$}, and $\iS4$ be the systems  obtain by dropping the \textit{reductio ab absurdum}  rule, $\bot_c$, from 
$\nKK$, $\nT$, $\nD$, $\nK4$, \rosso{$\nDQ$}, and $\nS4$, respectively.

We write 
$$\vertil{
          \pf{B}{\beta} \\ \Pi \\ \pf{A}{\alpha}}$$  
\noindent to say that $\Pi$ is a 
deduction of
        $\pf{A}{\alpha}$ having some {(possibly zero)} occurrences of formula
        $\pf{B}{\beta}$ among its assumptions, and we write
        
$$\ded{\Pi}{R}{\pf{A}{\alpha}}$$ 
to say that $\Pi$ is a deduction
   of formula $\pf{A}{\alpha}$ whose last rule is $R$.

To define the \textit{normal form} for a deduction, we must first introduce the notions
of \textit{contractions}, \textit{reduction steps}, and \textit{reduction sequence} (see, e.g.,~\cite{Girard:ptlc}.) 

\subsection{Proper contractions}
The relation $\rhd$ of \textit{proper contractibility}  between 
deductions is defined as follows.\footnote{Since the conclusion of $\bot_i$ is always atomic, we do not have contractions  associated to such a rule.}

\bigskip
\begin{description}
	\item[Proper contractibility for $\iT$, $\iD$, \rosso{$\iDQ$}, $\iS4$ systems] \mbox{}\\\ \\
	%%%%%
	\def\uno{\DT{}{{\Pi}_1}{\pf{A}{\alpha}}}
	\def\due{\DT{}{{\Pi}_2}{\pf{B}{\alpha}}}
	\def\tre{\prova{\uno\quad\due \ }{\pf{A\land B}{\alpha}}{ }}
	\def\quattro{\prova{\tre}{\pf{A}{\alpha}}{ }}
	\def\cinque{\prova{\tre}{\pf{B}{\alpha}}{ }}
	\def\sei{\prova{\uno\ }{\pf{A\lor B}{\alpha}}{ }}
	\def\sette{\DT{[\pf{A}{\alpha}]}{{\Pi}_2}{\pf{C}{\beta}}}
	\def\otto{\DT{[\pf{B}{\alpha}]}{{\Pi}_3}{\pf{C}{\beta}}}
	\def\nove{\prova{\sei\ \sette\ \otto\ }{\pf{C}{\beta}}{ }}
	\def\dieci{\DT{\uno}{{\Pi}_2}{\pf{C}{\beta}}}
	\def\duebis{\DT{}{{\Pi}_1}{\pf{B}{\alpha}}}
	\def\undici{\prova{\duebis}{\pf{A\lor B}{\alpha}}{ }}
	\def\dodici{\prova{\undici\ \sette\ \otto\ }{\pf{C}{\beta}}{ }}
	\def\tredici{\DT{\duebis}{{\Pi}_3}{\pf{C}{\beta}}}
	\def\quattordici{\DT{[\pf{A}{\alpha}]}{{\Pi}_1}{\pf{B}{\alpha}}}
	\def\quindici{\prova{\quattordici}{\pf{A\to B}{\alpha}}{ }}
	\def\sedici{\DT{}{{\Pi}_2}{\pf{A}{\alpha}}}
	\def\diciassette{\prova{\quindici\ \sedici}{\pf{B}{\alpha}}{ }}
	\def\diciotto{\DT{\sedici}{{\Pi}_1}{\pf{B}{\alpha}}}
	\def\venti{\DT{}{\Pi}{\pf{A}{s\oplus 1}}}
	\def\ventuno{\prova{\venti}{\pf{ A}{\alpha}}{}}
	\def\ventidue{\prova{\ventuno}{\pf{A}{s\oplus 1}}{}}
	\def\ventitre{\DT{}{\Pi}{\pf{A}{\alpha \conc x}}}
	\def\ventiquattro{\prova{\ventitre}{\pf{\Box A}{\alpha}}{ }}
	\def\venticinque{\prova{\ventiquattro}{\pf{A}{\alpha \conc \beta}}{ }}
	\def\ventisei{\DT{}{{\Pi}\rep{\alpha\conc x}{\alpha\conc \beta}}{\pf{A}{\alpha \conc \beta}}}
	\def\ventisette{\DT{}{{\Pi}_1}{\pf{A}{\alpha \conc \beta}}}
	\def\ventotto{\prova{\ventisette}{\pf{\Diamond A}{\alpha}}{ }}
	\def\ventinove{\DT{[\pf{A}{\alpha \conc x}]}{{\Pi}_2}{\pf{C}{\gamma}}}
	\def\trenta{\prova{\ventotto\ \ventinove}{\pf{C}{\gamma}}{ }}
	\def\trentuno{\DT{}{{\Pi}_1}{\pf{A}{\alpha \conc \beta}}}
	\def\trentadue{\DT{\trentuno}{{\Pi}_2\rep{\alpha\conc x}{\alpha\conc\beta}}{\pf{C}{\gamma}}}
	\def\trentatre{\DT{}{{\Pi}_1}{\pf{A}{\alpha}}}
	\def\trentaquattro{\DT{[\pf{A}{\alpha \conc x}]}{{\Pi}_2}{\pf{A}{s\oplus 
				<1,x>}}}
	\def\trentacinque{\prova{\trentatre\ \trentaquattro}{\pf{A}{\alpha}}{}}
	\def\trentacinquebis{\prova{\trentatre\
			\trentaquattro}{\pf{A}{s\oplus(t\oplus 1)}}{}}
	\def\trentasei{\prova{\trentatre\
			\trentaquattro}{\pf{A}{\alpha \conc \beta}}{}}
	\def\trentasette{\DT{\trentasei}{\Pi_2\rep{\beta}{\gamma}}{\pf{A}{s\oplus(t\oplus 1)}}}
	\begin{tabular}{ll}
		$ \quattro \quad \rhd\quad  \uno $
		&\hspace{2ex}
		$ \cinque  \quad\rhd \quad \due  $
		\\ & \\
		$\nove  \quad\rhd \quad \dieci$
		&\hspace{2ex}
		$\dodici \quad\rhd \quad \tredici $
		\\ & \\
		$\diciassette  \quad\rhd \quad\diciotto $
		&\hspace{2ex}
		\\ & \\
		$\venticinque \quad\rhd \quad\ventisei$
		&\hspace{2ex}
		$\trenta  \quad\rhd \quad\trentadue $
	\end{tabular}
	
	\vspace{1ex}
	
	\item[Proper contractibility for $\iKK$, $\iK4$ systems] 
	The same propositional contractions of the previous systems; the modal ones are adapted as follows. 
	 \mbox{}\\\ \\
	%%%%%
	\def\uno{\DT{}{{\Pi}_1}{\pf{A}{\alpha}}}
	\def\due{\DT{}{{\Pi}_2}{\pf{B}{\alpha}}}
	\def\tre{\prova{\uno\quad\due \ }{\pf{A\land B}{\alpha}}{ }}
	\def\quattro{\prova{\tre}{\pf{A}{\alpha}}{ }}
	\def\cinque{\prova{\tre}{\pf{B}{\alpha}}{ }}
	\def\sei{\prova{\uno\ }{\pf{A\lor B}{\alpha}}{ }}
	\def\sette{\DT{[\pf{A}{\alpha}]}{{\Pi}_2}{\pf{C}{\beta}}}
	\def\otto{\DT{[\pf{B}{\alpha}]}{{\Pi}_3}{\pf{C}{\beta}}}
	\def\nove{\prova{\sei\ \sette\ \otto\ }{\pf{C}{\beta}}{ }}
	\def\dieci{\DT{\uno}{{\Pi}_2}{\pf{C}{\beta}}}
	\def\duebis{\DT{}{{\Pi}_1}{\pf{B}{\alpha}}}
	\def\undici{\prova{\duebis}{\pf{A\lor B}{\alpha}}{ }}
	\def\dodici{\prova{\undici\ \sette\ \otto\ }{\pf{C}{\beta}}{ }}
	\def\tredici{\DT{\duebis}{{\Pi}_3}{\pf{C}{\beta}}}
	\def\quattordici{\DT{[\pf{A}{\alpha}]}{{\Pi}_1}{\pf{B}{\alpha}}}
	\def\quindici{\prova{\quattordici}{\pf{A\to B}{\alpha}}{ }}
	\def\sedici{\DT{}{{\Pi}_2}{\pf{A}{\alpha}}}
	\def\diciassette{\prova{\quindici\ \sedici}{\pf{B}{\alpha}}{ }}
	\def\diciotto{\DT{\sedici}{{\Pi}_1}{\pf{B}{\alpha}}}
	\def\venti{\DT{}{\Pi}{\pf{A}{s\oplus 1}}}
	\def\ventuno{\prova{\venti}{\pf{ A}{\alpha}}{}}
	\def\ventidue{\prova{\ventuno}{\pf{A}{s\oplus 1}}{}}
\def\ventitre{\DT{[\ext{\alpha\conc x}]}{\Pi}{\pf{A}{\alpha \conc x}  }}
\def\ventiquattro{\prova{\ventitre}{\pf{\Box A}{\alpha}}{ } \ext{\alpha\conc\beta}  }
	\def\venticinque{\prova{\ventiquattro}{\pf{A}{\alpha \conc \beta}}{ }}
\def\ventisei{\DT{\ext{\alpha\conc\beta}}{{\Pi}\rep{\alpha\conc x}{\alpha\conc \beta}}{\pf{A}{\alpha \conc \beta}}}
	\def\ventisette{\DT{}{{\Pi}_1}{\pf{A}{\alpha \conc \beta}}}
\def\ventotto{\prova{\ventisette \ext{\alpha\conc \beta}}{\pf{\Diamond A}{\alpha}}{ }}
\def\ventinove{\DT{[\pf{A}{\alpha \conc x}] \quad [\ext{\alpha\conc x}]}{\!{\Pi}_2}{\pf{C}{\gamma}}}	\def\trenta{\prova{\ventotto\ \ventinove}{\pf{C}{\gamma}}{ }}
	\def\trentuno{\DT{}{{\Pi}_1}{\pf{A}{\alpha \conc \beta}}}
\def\trentadue{\DT{\trentuno  
		\ext{\alpha\conc\beta}
	}{\!\!\!{\Pi}_2\rep{\alpha\conc x}{\alpha\conc\beta}}{\pf{C}{\gamma}}}
	\def\trentatre{\DT{}{{\Pi}_1}{\pf{A}{\alpha}}}
	\def\trentaquattro{\DT{[\pf{A}{\alpha \conc x}]}{{\Pi}_2}{\pf{A}{s\oplus 
				<1,x>}}}
	\def\trentacinque{\prova{\trentatre\ \trentaquattro}{\pf{A}{\alpha}}{}}
	\def\trentacinquebis{\prova{\trentatre\
			\trentaquattro}{\pf{A}{s\oplus(t\oplus 1)}}{}}
	\def\trentasei{\prova{\trentatre\
			\trentaquattro}{\pf{A}{\alpha \conc \beta}}{}}
	\def\trentasette{\DT{\trentasei}{\Pi_2\rep{\beta}{\gamma}}{\pf{A}{s\oplus(t\oplus 1)}}}
	\mbox{\\}
	\bigskip
		{$\venticinque \quad\rhd \quad\ventisei$}
		\\
		$\trenta  \quad\rhd \quad\trentadue $

\end{description}

\subsection{Commutative contractions}
	In this subsection, we denote by 
$$\centerline{\raisebox{4ex}{$\prova{\DT{}{\Pi_1}{\pf{C}{\beta}}\quad
			\Pi_2}{\pf{D}{\gamma}}{\ R}$}}$$
\noindent   a deduction ending with an  elimination 
rule $R$ whose major premiss  is formula
$\pf{C}{\beta}$.
	We further extend the relation $\rhd$ by adding the
following \textit{commutative contractions}:
\begin{description}
	\item[Commutative contractions for $\iT$, $\iD$, \rosso{$\iDQ$}, $\iS4$ systems] \mbox{}\\\ \\
	\newcommand{\sureg}[2]{\prova{ #1\quad {{\Pi}_{#2}}}{\pf{D}{\gamma}}{\ R}}
	\newcommand{\surego}[2]{\prova{ #1\quad {{\Pi}_{#2}}}{\pf{A}{\alpha}}{\ R}}
	\def\quaranta{\DT{}{{\Pi}_1}{\pf{A\lor B}{\alpha}}}
	\def\quarantuno{\DT{[\pf{A}{\alpha}]}{{\Pi}_2}{\pf{C}{\beta}}}
	\def\quarantadue{\DT{[\pf{B}{\alpha}]}{{\Pi}_3}{\pf{C}{\beta}}}
	\def\quarantatre{\prova{\quaranta\ \quarantuno\ 
			\quarantadue}{\pf{C}{\beta}}{}}
	\def\quarantaquattro{\sureg{\quarantatre}{4}}
	\def\quarantacinque{\sureg{\quarantuno}{4}}
	\def\quarantasei{\sureg{\quarantadue}{4}}
	\def\quarantasette{\prova{\quaranta\ \quarantacinque\
			\quarantasei}{\pf{D}{\gamma}}{ }}
	$$\quarantaquattro  \quad\rhd\quad \quarantasette$$
	\def\cinquanta{\DT{}{{\Pi}_1}{\pf{\Diamond A}{\alpha}}}
	\def\cinquantuno{\DT{[\pf{A}{\alpha \conc x}]}{{\Pi}_2}{\pf{C}{\beta}}}
	\def\cinquantadue{\prova{\cinquanta\ \cinquantuno}{\pf{C}{\beta}}{}}
	\def\cinquantatre{\sureg{\cinquantadue}{3}}
	\def\cinquantaquattro{\sureg{\cinquantuno}{3}}
	\def\cinquantacinque{\prova{\cinquanta\ \cinquantaquattro}{\pf{D}{\gamma}}{}}
	\def\cento{\DT{}{{\Pi}_1}{\pf{\Diamond B}{t}}}
	\def\centouno{\DT{[\pf{B}{t\oplus x}]}{{\Pi}_2}{\pf{C}{u}}}
	\def\centodue{\prova{\cento\ \centouno}{\pf{C}{u}}{}}
	\def\centotre{\surego{\centodue}{3}}
	\def\centoquattro{\surego{\centouno}{3}}
	\def\centocinque{\prova{\cento\ \centoquattro}{\pf{A}{\alpha}}{}}
	\def\duecento{\DT{}{{\Pi'}_1}{\pf{\Diamond B}{t}}}
	\def\duecentodue{\prova{\duecento\ \centouno}{\pf{C}{u}}{}}
	$$\cinquantatre \quad\rhd\quad \cinquantacinque$$
	%%%%%%%%%%%%%%%%%%%%%%%%%%%%%%%%%%%%%%%%%%
	\item[Commutative contractions for $\iKK$, $\iK4$ systems] 
	The same propositional commutative contractions of the previous systems; the modal ones are adapted as follows. 
	\mbox{}\\\ \\
	\def\quaranta{\DT{}{{\Pi}_1}{\pf{A\lor B}{\alpha}}}
	\def\quarantuno{\DT{[\pf{A}{\alpha}]}{{\Pi}_2}{\pf{C}{\beta}}}
	\def\quarantadue{\DT{[\pf{B}{\alpha}]}{{\Pi}_3}{\pf{C}{\beta}}}
	\def\quarantatre{\prova{\quaranta\ \quarantuno\ 
			\quarantadue}{\pf{C}{\beta}}{}}
	\def\quarantaquattro{\sureg{\quarantatre}{4}}
	\def\quarantacinque{\sureg{\quarantuno}{4}}
	\def\quarantasei{\sureg{\quarantadue}{4}}
	\def\quarantasette{\prova{\quaranta\ \quarantacinque\
			\quarantasei}{\pf{D}{\gamma}}{ }}
	\def\cinquanta{\DT{}{{\Pi}_1}{\pf{\Diamond A}{\alpha}}}
	\def\cinquantuno{\DT{[\pf{A}{\alpha \conc x}]\quad 
			[\ext{\alpha\conc x}]}{{\Pi}_2}{\pf{C}{\beta}}}
	\def\cinquantadue{\prova{\cinquanta\ \cinquantuno}{\pf{C}{\beta}}{}}
	\def\cinquantatre{\sureg{\cinquantadue}{3}}
	\def\cinquantaquattro{\sureg{\cinquantuno}{3}}
	\def\cinquantacinque{\prova{\cinquanta\ \cinquantaquattro}{\pf{D}{\gamma}}{}}
	\def\cento{\DT{}{{\Pi}_1}{\pf{\Diamond B}{t}}}
	\def\centouno{\DT{[\pf{B}{t\oplus x}]}{{\Pi}_2}{\pf{C}{u}}}
	\def\centodue{\prova{\cento\ \centouno}{\pf{C}{u}}{}}
	\def\centotre{\surego{\centodue}{3}}
	\def\centoquattro{\surego{\centouno}{3}}
	\def\centocinque{\prova{\cento\ \centoquattro}{\pf{A}{\alpha}}{}}
	\def\duecento{\DT{}{{\Pi'}_1}{\pf{\Diamond B}{t}}}
	\def\duecentodue{\prova{\duecento\ \centouno}{\pf{C}{u}}{}}
	$$\cinquantatre \quad\rhd\quad \cinquantacinque$$
	
\end{description}

\begin{remark} It is easy to verify that  contractions transform 
deductions into deductions.  Furthermore, they all preserve the position 
condition.
\end{remark}

\begin{definition}[Reducibility between Deductions]\label{def:red}
  \begin{enumerate}\mbox{}
\item  The relation $\red$ of \textit{immediate reducibility}  between 
deductions is the ``context closure'' of $\rhd$, defined
   as follows: $\Pi_1\red \Pi_2$ if and only if there exist deductions  
$\Pi_3$ and
   $\Pi_4$ such that $\Pi_3\rhd\Pi_4$ and $\Pi_2$ is obtained by
   replacing $\Pi_3$ with $\Pi_4$ in $\Pi_1$.
  \item The relation  $\rred$ of \textit{reducibility} is the transitive 
and reflexive closure of $\red$.
\end{enumerate}
\end{definition}

\subsection{Normalization}
The results of the following section apply to all the previously introduced, intuitionistic systems.

\begin{definition}[Normal forms and normalizable deductions]\label{def:nfnorm} A deduction $\Pi$ is
\begin{enumerate}
\item in \textit{normal form}
   if there is no deduction $\Pi'$ such that $\Pi\red\Pi'$;
\item \textit{normalizable} if there is a deduction $\Pi'$ s.t. $\Pi\rred \Pi'$ and $\Pi'$ is in normal form.
\end{enumerate}
\end{definition}

\begin{definition}[Segments and Endsegments]\label{def:seg}
	Let $\pf{A}{\alpha}$ be a p-formula.
	\begin{enumerate}
		\item A finite sequence $({\pf{A}{\alpha}}_i)_{i\leq m} $ of  occurrences of  $\pf{A}{\alpha}$ in a
		deduction $\Pi$ is a \textit{segment} (of length $m+1$) if:
		\begin{enumerate}
			\item ${\pf{A}{\alpha}}_0$ is  not a conclusion of $\lor E$ or $\Diamond E$;
			\item  ${\pf{A}{\alpha}}_m$ is not a minor premiss of $\lor E$ or $ \Diamond 
			E$;
			\item for all $i \lt m$, ${\pf{A}{\alpha}}_i$ is a minor premiss of $\lor E$
			or $ \Diamond E$ with conclusion ${\pf{A}{\alpha}}_{i+1}$
		\end{enumerate}		
		\item A segment in a deduction is an \textit{endsegment} if its last 
		formula is the last formula of  the deduction.
	\end{enumerate}
	
\end{definition}

	We will denote segments with $\sigma$, possibly indexed.
	When we want to highlight  that a segment is made of occurrences of a formula $\pf{A}{\alpha}$ we will write $\sigma[{\pf{A}{\alpha}}]$. With $|\sigma|$ we denote the length of the segment $\sigma$.

Given a deduction $\ded{\Pi}{R}{\pf{A}{\alpha}},$ with little abuse of 
language  we will say that a deduction $\Pi'$
   is a  \textit{(main) premiss of rule} $R$ to mean that
$\Pi'$ is a sub-deduction of $\Pi$ whose end-formula is a (main)  premiss 
of the displayed application of $R.$

\begin{definition}[Degree of a formula]\label{def:degree}
	\begin{enumerate}
		\item
		The \textit{degree}  $\deg(A)$  of a modal formula $A$  is
		recursively  defined as:
		
		\begin{enumerate}
			\item $\deg(p) = 0$ if $p$ is a proposition symbol;
			\item $\deg(\lnot A) = \deg(\Box A) = \deg(\Diamond A) = \deg(A)+1$;
			\item $\deg(A \land B) = \deg(A\lor B) = \deg(A\to B) = \max\{\deg(A),\deg(B)\}+1$.
		\end{enumerate}
		\item The \textit{degree} $\deg(\pf{A}{\alpha})$ of formula $\pf{A}{\alpha}$ is just $\deg(A).$
	\end{enumerate}
\end{definition}

\begin{definition}[Major/Minor Premisses and Conclusions]
	Let $\sigma[{\pf{A}{\alpha}}]={\pf{A}{\alpha}}_0 \ldots \pf{A}{\alpha}_m$
	and let $R$ be  a segment and an instance of a deduction rule  in  $\Pi$, respectively. We say that:
	\begin{itemize}
		\item $\sigma$ is the  \textit{(major/minor) premiss} of $R$, if $ {\pf{A}{\alpha}}_m$ is  the  (major/minor) premiss of $R$;
		\item $\sigma$ is  \textit{conclusion} of $R$, if $ {\pf{A}{\alpha}}_0$ is  the  conclusion of $R$.
	\end{itemize}
\end{definition}

With $\delta(\sigma[\pf{A}{\alpha}])= d(A)$ we denote the \textit{degree} of the segment $\sigma[\pf{A}{\alpha}]$.

\begin{definition}[cut] 
	\begin{enumerate}
		\item A \textit{cut} in a derivation $\Pi$ is a segment $\sigma$  which is conclusion of an introduction rule $I*$  of a connective $*$, and principal premiss of an elimination rule $E*$ of the same connective.
		\item A  cut $\sigma$ in $\Pi$ is  \textit{maximal}  if $\delta(\sigma)=\max\{\delta(\sigma') : \sigma' \mbox{\ is a cut in\ } \Pi \}$.
		\item A \textit{(maximal) cut formula} is a (maximal) cut segment of length 1.
	\end{enumerate}
\end{definition}

Let $C[\Pi]$ be the set of  cuts  of $\Pi$. 
For the normalization theorem we will use the lexicographic ordering between pairs of natural numbers\footnote{ 
$(n,m)\lt (p,q)$ if either $n\lt p$ or ($n=p$ and $m\lt q$) }.

\begin{theorem}[normalization]
For each derivation $\Pi$ there exists a derivation $\Pi'$ s.t. $\Pi\rred \Pi'$ and $\Pi'$ is in normal form.
\end{theorem}
\begin{proof}
The proof is on well ordering induction on pairs $(d,n)$ of natural numbers.
We associate to each derivation $\Pi$ a pair (called \textit{rank}) $\#[\Pi]=(d,n)$ s.t.
\begin{itemize}
	\item $d= max\{ \delta(\sigma): \sigma\in C[\Pi]  \}$;
	\item $n=\sum_{\sigma\in C[\Pi], \delta(\sigma)=d} |\sigma|$. 
\end{itemize}
We then prove the following \textbf{claim}:
$$
\#[\Pi] \gt (0,0) \Rightarrow \exists \Pi'(\Pi \rred \Pi'\ \& \  \#[\Pi'] \lt \#[\Pi])
.
$$
\begin{enumerate}
	\item Let us suppose that $\#[\Pi] \gt (0,0) $;
	\item pick a maximal cut $\sigma$ in $\Pi$ s.t. the sub-derivation $\Pi*$ ending with $\sigma$ (i.e. ending with the last occurrence of $\sigma$ ) does not contain any other maximal cut segment;
	\item perform all possible {commutative contractions} with respect to the segment under consideration;
	\item perform the relevant contraction.
\end{enumerate}
The resulting derivation $\Pi'$ has a smaller rank w.r.t $\Pi$ i.e. $\#[\Pi'] \lt \#[\Pi]$.

Using the \textbf{claim}, since the lexicographic order is well founded, for each derivation $\Pi$ there exists a derivation $\Pi'$ s.t. $\Pi\rred \Pi'$ and $\#[\Pi']=(0,0)$, i.e. the thesis.
\end{proof}

\section{Consequences of normalization}

Let us denote with $\class$ one of the previously stated classical systems, and with $\int$ the corresponding intuitionistic system.

\begin{definition}[Spine]\label{spinedef}
	A finite sequence $(\pf{A_i}{\alpha_i})_{i\leq m} $ of formulas in a
	deduction is a \textit{spine} if:
	\begin{enumerate}
		\item for all $i\lt m$, $\pf{A_i}{\alpha_i}$ is immediately above
		$\pf{A_{i+1}}{s_{i+1}}$;
		\item $\pf{A_m}{\alpha_m}$ is the end-formula  of the deduction;
		\item $\pf{A_0}{\alpha_0}$ is an assumption (either discharged or undischarged);
		\item for all $i \lt m$, $\pf{A_i}{\alpha_i}$ is one of the following:
		\begin{enumerate}
			\item main premiss of some  elimination rule;
			\item premiss of some  introduction rule;
			\item premiss of an application of  $\bot_i$ rule.
		\end{enumerate}
	\end{enumerate}
\end{definition}

Spines in normal deductions have a nice structure.  It is easy to prove the
following:

\begin{proposition}\label{spinestr} A  spine $(\pf{A_i}{\alpha_i})_{i\leq n}$ in a normal deduction 
	can be divided into three subsequences:
	\begin{enumerate}
		
		\item an elimination sequence $(\pf{A_i}{\alpha_i})_{i\leq m}$ where each 
		$\pf{A_i}{\alpha_i},$ $i\lt m,$ is main premiss of some elimination rule;
		
		\item a minimum sequence $(\pf{A_i}{s_i})_{m\lt  i\leq m+k}$ where each 
		$\pf{A_i}{\alpha_i},$ $m\lt  i\lt  m+k$ is premiss of $\bot_i;$
		
		\item an introduction  sequence $(\pf{A_i}{\alpha_i})_{m+k\leq i\leq n}$ where each 
		$\pf{A_i}{\alpha_i},$ $m+k\lt  i\lt n$ is  premiss of some introduction rule.
	\end{enumerate}
	
	In particular, in a normal deduction whose last rule is not an
	introduction there is a unique spine. The spine does not contain the introduction  sequence.

\end{proposition}

As an immediate consequence we have the following Consistency Theorem:
\begin{theorem}[Consistency]
For each position $\alpha$, $\not\vdash_{\int} \pf{\bot}{\alpha} $.
\end{theorem}

\subsection{A translation of the classical calculi into the intuitionistic ones}\label{Sect:doubleneg}
To obtain a syntactical proof of consistency for the classical systems, we adapt G{\"o}del's {\em double negation translation} to our setting.
As usual, $A\leftrightarrow B$ is an abbreviation for $(A\to
B)\land(B\to A).$

We inductively define a map $g$ between modal formulas as follows:
\begin{itemize}
	\item[] $g(\bot)=\bot;$
	\item[] $g(A)=\lnot\lnot A$ for atomic $A$ distinct from $\bot;$
	\item[] $g(A\lor B)=\lnot(\lnot g(A)\land \lnot g(B));$
	\item[] $g(A\sharp B)= g(A)\,\sharp\, g(B)$ when $\sharp$ is a binary connective
	distinct from $\lor;$
	\item[] $g(\Box A)=\Box g(A);$
	\item[] $g(\Diamond A)=\lnot\Box\lnot g(A);$
\end{itemize}

\begin{proposition}\label{aga} 
	For every modal formula $A$ and every position $\alpha$, 
	$$\vdash_{\class}\pf{(A\leftrightarrow g(A))}{\alpha}.$$
\end{proposition}

\begin{definition}[Negative Formulas] A modal formula is {\em negative} if it is constructed
	from $\bot$ or from atomic formulas by means of $\Box$, $\land$, 
	$\to$.
\end{definition}

\begin{lemma}\label{notnot} Let $A$ be a negative formula  constructed from doubly negated atomic formulas
	or  from $\bot.$ Then, for all positions $\alpha$ 
	$$\vdash_{\int}\pf{(A\leftrightarrow
		\lnot\lnot A)}{\alpha}.$$
\end{lemma}

\begin{proof}By induction on the complexity of $A.$
	
	\begin{itemize}
\item 	For the basis, recall that if  $A$ is either $\bot$ or a
	doubly negated atomic formula  then $A$ is provably equivalent to $\lnot\lnot A$
	in an intuitionistic framework.
	
\item	Concerning the induction step, we only  examine some nontrivial cases.
	\begin{description}
		\item[$\Box A$:] Suppose the statement true for $A.$ Then 
		$$\vdash_{\int}\pf{(\Box A\leftrightarrow\Box\lnot\lnot A)}{\alpha}$$ for all
		positions $s.$ Therefore, to prove the
		nontrivial implication $\vdash_{\int}\pf{(\lnot\lnot\Box A\to\Box
			A)}{\alpha},$ it suffices to show that 
		$\vdash_{\int}\pf{(\lnot\lnot\Box A\to\Box\lnot\lnot A)}{\alpha}.$
		The latter holds since
		$$\vdash_{\int}\pf{(\Diamond\lnot
			A\to\lnot\Box A)}{\alpha}\quad\mbox{and}\quad
		\vdash_{\int}\pf{(\lnot\Diamond\lnot A\to\Box\lnot\lnot A)}{\alpha}$$ are true
		for all positions $s,$ even with no assumption on $A.$
		
		\item[$A\to B$:] Suppose $\vdash_{\int}\pf{(B\leftrightarrow\lnot\lnot
			B)}{\alpha}.$  Then
		$\vdash_{\int}\pf{(\lnot\lnot(A\to B)\leftrightarrow A\to\lnot\lnot
			B)}{\alpha}$  and\\  $\vdash_{\int}\pf{(A\to \lnot\lnot B\leftrightarrow A\to
			B)}{\alpha}.$  Hence 
		$$\vdash_{\int}\pf{(A\to B\leftrightarrow \lnot\lnot(A\to
			B))}{\alpha}$$  for all positions $s.$
	\end{description}
	\end{itemize}
\end{proof}

\begin{remark}\label{oknotnot} For every modal formula $A,$ the formula $g(A)$ satisfies the
	assumptions of Lemma~\ref{notnot}.
\end{remark}

\begin{remark}\label{contronom} The following holds for any set  $\Gamma$ of formulas and 
	formulas $\pf{A}{\alpha}$ and $\pf{B}{\beta}$:  if $\Gamma,\pf{A}{\alpha}\vdash_{\int}\pf{B}{\beta}$ then 
	$\Gamma,\pf{\lnot B}{\beta}\vdash_{\int}\pf{\lnot A}{\alpha}.$
\end{remark}

We can now prove the following:

\begin{proposition} For every family $\{\pf{B_i}{\alpha_i}:\, i\in I\}$ of formulas and
	every formula  $\pf{A}{\alpha}$
	$$\{\pf{B_i}{\alpha_i}:\, i\in I\}\vdash_{\class}\pf{A}{\alpha}\ \Leftrightarrow\ \{\pf{g(B_i)}{\alpha_i}:\, i\in I\}\vdash_{\int}\pf{g(A)}{\alpha}.$$
\end{proposition} 

\begin{proof}
	\begin{itemize}
		\item[$(\Leftarrow)$] Straightforward from  Remark~\ref{aga}.
		\item[$(\Rightarrow)$] By induction on the height  of a deduction of $\pf{A}{\alpha}$
		in $\class$. We only examine some  nontrivial cases of  the induction  step.
		\begin{itemize}
			\item[($\Diamond E$)] Suppose \\

			\begin{prooftree}
				\[\dots\  \pf{B_i}{\alpha_i}\  \dots
				\using                          
				{}                       
				\proofdotseparation=1.2ex       
				\proofdotnumber=4               
				\leadsto
				\pf{\Diamond C}{\beta}\]
				\[[\pf{C}{\beta x}]\ \dots \pf{B_i}{\alpha_i} \dots
				\using                          
				{}                       
				\proofdotseparation=1.2ex       
				\proofdotnumber=4               
				\leadsto
				\pf{A}{\alpha}\]
				\justifies
				\pf{A}{\alpha}
				\thickness=0.06em
				\using
				{}                                      
			\end{prooftree}\\
			
	  in $\class$. Then (inductively) we get  the
			deductions \\\medskip
			
		\begin{prooftree}
				\dots\  \pf{g(B_i)}{\alpha_i}\  \dots
				\using                          
				{}                       
				\proofdotseparation=1.2ex       
				\proofdotnumber=4               
				\leadsto
				\pf{\lnot\Box\lnot g(C)}{\beta}
			\end{prooftree}\quad and\quad
			\begin{prooftree}\pf{g(C)}{\beta  x}\ \dots \pf{g(B_i)}{\alpha_i} \dots
				\using                          
				{}                       
				\proofdotseparation=1.2ex       
				\proofdotnumber=4               
				\leadsto
				\pf{g(A)}{\alpha}
			\end{prooftree}\medskip

			in $\int$. By Remark~\ref{contronom}, Remark~\ref{oknotnot}  and Lemma~\ref{notnot} we get
			the following deduction in $\int$\  (we leave to the reader to check that all side
			conditions of deduction rules are fulfilled):
			
			\bigskip
			
			\centerline{
				\begin{prooftree}
					\[\[\[\[\[[\pf{\lnot g(A)}{\alpha}]\ \dots \pf{g(B_i)}{\alpha_i} \dots
					\using                          
					{}                       
					\proofdotseparation=1.2ex       
					\proofdotnumber=4               
					\leadsto
					\pf{\lnot g(C)}{\beta x}\]
					\justifies
					\pf{\Box\lnot g(C)}{\beta}
					\thickness=0.06em
					\using
					{}\]
					\[\dots\  \pf{g(B_i)}{\alpha_i}\  \dots
					\using                          
					{}                       
					\proofdotseparation=1.2ex       
					\proofdotnumber=4               
					\leadsto
					\pf{\lnot\Box\lnot g(C)}{\beta}\]
					\justifies
					\pf{\bot}{\beta}
					\thickness=0.06em
					\using
					{}\]
					\justifies
					\pf{\bot}{\alpha}
					\thickness=0.06em
					\using
					{}\]
					\justifies
					\pf{\lnot\lnot g(A)}{\alpha}
					\thickness=0.06em
					\using
					{}\]
					\[{}
					\using
					{}
					\proofdotseparation=1.2ex
					\proofdotnumber=4               
					\leadsto
					\pf{\lnot\lnot g(A)\to g(A)}{\alpha}
					\]
					\justifies
					\pf{g(A)}{\alpha}
					\thickness=0.06em
					\using
					{} 
			\end{prooftree}}
			
			%%%%%%%%%%%%%
			%%%%%%%%%%%%%
			\bigskip
			
			\item[($\lor E$)] Suppose\\
			
			\begin{prooftree}
				\[\dots\pf{B_i}{\alpha_i}\dots
				\using                          
				{}                       
				\proofdotseparation=1.2ex       
				\proofdotnumber=4               
				\leadsto
				\pf{B\lor C}{\beta} \]
				\[[\pf{B}{\beta}]\ \dots\pf{B_i}{\alpha_i}\dots
				\using                          
				{}                       
				\proofdotseparation=1.2ex       
				\proofdotnumber=4               
				\leadsto \pf{A}{\alpha}\]
				\[[\pf{C}{\beta}]\ \dots\pf{B_i}{\alpha_i}\dots
				\using                          
				{}                       
				\proofdotseparation=1.2ex       
				\proofdotnumber=4               
				\leadsto \pf{A}{\alpha}\]
				\justifies 
				\pf{A}{\alpha}
				\thickness=0.06em
				\using
				{}
			\end{prooftree}\quad in $\class$.
			
			\bigskip
			
			By induction hypothesis and by Remark~\ref{contronom} we get the following
			deductions in $\int$:
			
			\bigskip
			
			\begin{prooftree}
				\dots\pf{g(B_i)}{\alpha_i}\dots
				\using                          
				{}                       
				\proofdotseparation=1.2ex       
				\proofdotnumber=4               
				\leadsto
				\pf{\lnot(\lnot g(B)\land \lnot g(C))}{\beta}
			\end{prooftree}\qquad
			\begin{prooftree}
				\pf{\lnot g(A)}{\alpha}\ \dots\pf{g(B_i)}{\alpha_i}\dots
				\using                          
				{}                       
				\proofdotseparation=1.2ex       
				\proofdotnumber=4               
				\leadsto
				\pf{\lnot g(B)}{\beta}
			\end{prooftree}\qquad
			\begin{prooftree}
				\pf{\lnot g(A)}{\alpha}\ \dots\pf{g(B_i)}{\alpha_i}\dots
				\using                          
				{}                       
				\proofdotseparation=1.2ex       
				\proofdotnumber=4               
				\leadsto
				\pf{\lnot g(C)}{\beta}
			\end{prooftree}
			
			\bigskip
			
			From these deductions we can produce the following in $\int$:
			
			\bigskip
			
			\begin{prooftree}
				\[\[\[\dots\pf{g(B_i)}{\alpha_i}\dots
				\using                          
				{}                       
				\proofdotseparation=1.2ex       
				\proofdotnumber=4               
				\leadsto
				\pf{\lnot(\lnot g(B)\land \lnot g(C))}{\beta}\]
				\[
				\[[\pf{\lnot g(A)}{\alpha}]\ \dots\pf{g(B_i)}{\alpha_i}\dots
				\using                          
				{}                       
				\proofdotseparation=1.2ex       
				\proofdotnumber=4               
				\leadsto
				\pf{\lnot g(B)}{\beta} \]
				\[[\pf{\lnot g(A)}{\alpha}]\ \dots\pf{g(B_i)}{\alpha_i}\dots
				\using                          
				{}                       
				\proofdotseparation=1.2ex       
				\proofdotnumber=4               
				\leadsto
				\pf{\lnot g(C)}{\beta} \]
				\justifies
				\pf{\lnot g(B)\land\lnot g(C)}{\beta}
				\thickness=0.06em
				\using 
				{}
				\]
				\justifies
				\pf{\bot}{\beta}
				\thickness=0.06em
				\using 
				{}
				\]
				\justifies
				\pf{\bot}{\alpha}
				\thickness=0.06em
				\using 
				{}
				\]
				\justifies
				\pf{\lnot\lnot g(A)}{\alpha}
				\thickness=0.06em
				\using 
				{}
			\end{prooftree}
			
			\bigskip
			
			We finally get the required deduction in $\int$\  from Lemma~\ref{notnot}.
			
			\bigskip
			
			\noindent The other cases are easier.
			
		\end{itemize}
	\end{itemize}
\end{proof}

\begin{corollary}\label{trasl} For every formula $\pf{A}{\alpha}$
	$$\vdash_{\class}\pf{A}{\alpha}\ \Leftrightarrow\quad \vdash_{\int}\pf{g(A)}{\alpha}.$$
\end{corollary}

\begin{remark}\label{tnkcons} 
	Consistency of $\class$ follows immediately from Corollary~\ref{trasl}.
\end{remark}

% !TEX root = ./NatDedPos-revised.tex

\section{Discussions and future work}

In this paper, we defined natural deduction systems for normal modal logics, ranging from the basic {\KK} to S4. We have provided both the classical and the intuitionistic formulations. We followed the paradigm of 2-Sequents by Masini et al.~\cite{MM:ONFineStr:95,MM:ComInt:95,BaMaAML2004,BaMaJANCL13, BarMas:apal,Mas:TwoSeqProof:92,Mas:TwoSeqInt:93} and we introduced a notion of \emph{position} which represents the spatial coordinate of a formula.
For the intuitionistic versions of the systems, we defined proof reduction and proved proof normalization, thus obtaining a syntactical proof of consistency. We lifted the results of consistency to classical systems by adapting G\"odel's double negation translation.
 Natural deduction calculi for partial logics ({\KK} and {\K4}) are particularly challenging, and the sound formulation of the deduction system required the introduction of an existence predicate \`a la Scott~\cite{Scott1979}. 
We aimed to retain the original intention of natural deduction, as motivated by Prawitz~\cite{Prawitz:1965}.

In the following, we briefly discuss some crucial differences and analogies  between the framework we proposed and labelled deduction systems. Moreover, we sketch possible developments of our investigation. 

\subsection{Labelled natural deduction systems: a comparison}\label{sect:comparisons}
We start by recalling the basic elements of labelled systems, one of the most popular natural deduction formulations of modal logics. We focus on the original systems, as proposed by Simpson~\cite{Simpson93} and, later, by Vigan\`o~\cite{Vigano00a}, which are the roots of the approach.
They build on the well-known translation $(\cdot)^*_x$ that, given a propositional modal formula and a first-order variable $x$, produces a  first-order formula in a language  with denumerable many unary predicate symbols and one binary predicate symbol $\rl $ (which is going to be modeled by the accessibility relation in the Kripke model):

\begin{itemize}
	\item $(p_i)^*_x= P_i(x)$, where the $p_i$-s and $P_i$-s are the $i$-th propositional and the $i$-th predicate symbol, respectively;
	\item $(\bot)^*_x= \bot$ ;
	\item $(A\circ B)^*_x = (A)^*_x \circ (B)^*_x$, for each propositional connective $\circ$;
	\item $(\Box A)_x^* = \forall y ( x\rl y \to (A)^*_y )$,  for $y$ a fresh variable.
\end{itemize}

As a result of this translation, Simpson and Vigan\`o proposed natural deduction systems for a large class of modal logics, based on formulas for the accessibility relation (the \emph{relational formulas}), with explicit rules governing the properties of this relation.
The core rules, common for all normal modal logics, are listed in Figure~\ref{Fig:labelledrules} and are the same in the two approaches.  Systems for specific logics are obtained through a characteristic set of additional rules for the relational formulas. How these relational constraints are formulated and used in a derivation significantly differs in the two approaches. 

\begin{figure}[htb]
\begin{center}
   $ \urule { \LT{t:A}{[s\rl t] } } {s:{ \Box A} } {(\Box I)* } $ 
   \qquad  %
   $ \brule { \LT{s:\Box A}{ } }{s\rl t} { t:A} { (\Box E)
     } $
\end{center}

\bigskip

\begin{center}
   $ \brule { \LT{t:A}{ } }{s\rl t} { s:\Diamond A} { (\Diamond I) } $ 
   \qquad 
   $\brule{\LT{s:\Diamond A}{ }}
   {\LT{u:B}{[t:A][s\rl t]}} {u:B} {(\Diamond E)*} $
\end{center}

\rosso{
In $\Box I$, $t$ is not $s$ and does not occur in any assumption on which $t:A$ depends, other than $s\rl t$.\\
In $\Diamond E$, $t$ is neither $s$ or $u$, and does not occur in any assumption on which the upper occurrence of $u:B$ depends, other than $t:A$ and $s\rl t$.}
\caption{\rosso{Modal rules in labelled systems: logic \KK{}}}\label{Fig:labelledrules}
\end{figure}

\begin{description}
	\item[Simpson:]  \rosso{The additional rules for relational formulas act like structural rules; moreover,} any deduction must have a non-relational (thus modal, or propositional) formula as a conclusion, and the first-order relational formulas $x\rl y$ are used only as assumptions.
	\item[Vigan\`o:] \rosso{The additional rules for relational formulas axiomatize naturally the accessibility relation; moreover,} it is possible, using suitable rules, to built sub-derivations composed only by relational formulas.
\end{description}
Both approaches have strengths and weaknesses.
From a foundational point of view, Simpson's formulation is perhaps the most elegant, but it has the serious defect of making derivations complex (\emph{de facto}, not  natural at all). Take, for instance, Simpson's calculus for \K4, obtained by adding to Figure~\ref{Fig:labelledrules} the following rule
\begin{center}
 
   $ \trule {x\rl y}{y\rl z}{ \LT{w:A}{[x\rl z]} } { w:A} { (R_4)} $
\end{center}
The  ``structural'' rule $R4$ shows the price that this approach has to pay---the calculus includes explicit ``structural'' rules governing the accessibility relation. The following is the proof of formula {\bf{4}} in this system. 

\begin{center}
 
   $ 
   \urule{
     \urule{
       \urule{
         \trule{
         [s \rl t]^3
         }
         {[t \rl u]^2
         }
         {
           \brule{
           [s:\Box A]^4
           }
           {[s \rl u]^1
           }
           {u:A}{\Box E(1)}
         }
         {u:A}{R_4 (1)}
       }
       {t:\Box A}{\Box I (2)}
     }
     {s:\Box\Box A}{\Box I (3)}
   }
   {s:\Box A \to \Box\Box A}{\to I (4)}
    $
\end{center}

Moreover, to obtain normalization Simpson needs commutative reductions not only against $(\Diamond E)$, as usual, but also against rules with relational premises, see~\cite[Fig. 7-2, pag. 120]{Simpson93}.

On the other hand, Vigan\`o's formulation has the gift of simplicity. In particular, there are no structural rules in the system. However, it is a calculus that mirrors closely\footnote{Too closely, from our proof-theoretical perspective.} the first-order axiomatization of Kripke semantics. Vigano's calculus for \K4{} is obtained by adding to Figure~\ref{Fig:labelledrules} the following rule, consisting only of relational formulas:
\begin{center}
   $ \brule {s\rl t}{t\rl u}{s\rl u} { (\textit{trans})} $
\end{center}
The proof of the formula {\bf{4}} in this system becomes the following:
\begin{center}
   $ 
      \urule{
     \urule{
       \urule{
         \brule{
         [s:\Box A]^3
         }
         {
           \brule{
           [s \rl t]^2
           }
           {[t \rl u]^1
           }
           {s\rl u}{(\textit{trans})}
         }
         {u:A}{\Box E}
       }
       {t:\Box A}{\Box I (1)}
     }
     {s:\Box\Box A}{\Box I (2)}
   }
   {s:\Box A \to \Box\Box A}{\to I (3)}
   $
\end{center}
where a subproof consisting only of relational formulas has to be added. 

The strengths of these systems, on the other hand, become apparent when expressivity comes into the spotlight---both Simpson's  and Vigan\`o's proposals accommodate a large class of complex modal and temporal logics~\cite{Baratella2004a,Baratella2019,BaMaAML2004,MVVJANCL2101,MVVJLC2011,MasiniViganoZorzi08,MVVENTCS2010,MVZ-jmvl}, and have been successively formulated also as sequent calculi~\cite{Negri:2011}.

Contrary to these approaches, our central goal has been---as it should be clear, by now---to obtain a system with no structural rules, with rules only to introduce/eliminate logical connectives, whose modal rules are as close as possible to the first-order ones for the quantifiers, and with no explicit reference to the properties of the accessibility relation of the intended Kripke models. We have done so by internalizing (``hiding'') into positions the accessibility relation which labelled systems make explicit.  For the sake of clarity, we now
follow the inverse path, elaborating on the ``semantical'' interpretation of positions we provided in Section~\ref{sec:2seq}. In particular, we sketch how to extract from our framework a labelled system. For this, we consider the natural deduction version of (a fragment of) the labelled sequent system proposed by Negri in ~\cite{Negri:2011} (in its turn a variation of Simpson's natural deduction system~\cite{Simpson93} whose core rules have already been shown in Figure~\ref{Fig:labelledrules}.) The reference rules for  modalities are thus the following (for the sake of brevity, we give rules and derivations in linear style):

\bigskip

\begin{center}
   $\urule {\Gamma, s\rl t \vdash t:A}{\Gamma\vdash s:\Box A} 
     {\ ({\Box I}_L)*} $
\qquad $\urule{\Gamma\vdash s:\Box A}{\Gamma, s\rl t\vdash t:A}{\ (\Box E_L)}$ 
\end{center}
\bigskip
\begin{center}
$\urule{\Gamma, s\rl t \vdash t:A}{\Gamma, s\rl t \vdash s:\Diamond A}{\ (\Diamond I_L)}
\qquad
\brule{\Gamma \vdash s: \Diamond A}{\quad\Gamma, t:A, s\rl t \vdash u:B}{\Gamma \vdash u: B}{\ (\Diamond E_L)*}
$
\end{center}
\bigskip
\rosso{with the restrictions on labels for ${\Box I}_L$ and $\Diamond E_L$ stated in Figure~\ref{Fig:labelledrules}. Larger systems are obtained modularly, by adding some relational rules~\cite{Negri:2011}:}

\bigskip
\begin{center}
   \urule{\Gamma, t\rl s \vdash u:A}{\Gamma \vdash u: A}{\ (\textit{Ser})*}
\qquad $\urule{\Gamma t\rl t \vdash u:A}{\Gamma\vdash u:A}{ (\textit{Refl})}$ 
\qquad $\urule{\Gamma, t\rl s \vdash u:A}{\Gamma, t\rl v, v\rl s  \vdash u: A}{ (\textit{Trans})}$
\end{center}
\rosso{where, in (\textit{Ser}), $s$ is neither $t$ nor $u$.}\\

%\blue{Corretta, anche la formulazione sopra di $\Diamond I_L$}
%\green{****Ragazzi: va bene questa richiesta di fresheness? L'ho presa dalla tesi di Simpson.***}
%\bigskip

%\green{*********\\
%Ragazzi: ma la regola per $\Diamond I_L$ non dovrebbe avere la formula relazionale solo nella conclusione? Le due formulazioni sono forse inter-derivabili, via weakening (in una direzione) e contraction (nell'altra) sulla formula relazionale della premessa.}
%\begin{center}
%\green{
%$\urule{\Gamma \vdash t:A}{\Gamma, s\rl t \vdash s:\Diamond A}{\ (\Diamond I_L)}$
%}
%\end{center}
%\green{*********\\}

We stress that we are not looking after a full-blown translation between the two systems---we limit ourselves to sketch a procedure that extracts explicitly a labelled framework out of one of our systems, to show how positions internalize the ``structural'' rules of labelled systems.

We start by introducing a (new) label $\overline{\alpha}$ for each position (i.e., sequence of tokens) $\alpha$. Moreover, we make use of a binary predicate symbol $\rl$ between labels, to obtain relational formulas: if $\overline{\alpha}$ and $\overline{\beta}$ are labels, then $\overline{\alpha}\rl\overline{\beta}$ is a relational formula. We now associate to each position $\alpha$ a set of relational formulas:
$$\ptol{\langle\rangle}=\emptyset\qquad \ptol{x}=\{\overline{\langle\rangle}\rl \overline{x}\}\qquad \ptol{\alpha x}=\ptol{\alpha}\cup\{\overline{\alpha}\rl \overline{\alpha x}\}$$
Judgments, as defined in Section~\ref{sec:2seq},
can be translated into judgments of the labelled system:

\begin{equation} \label{eq:traduzione}
\ptol{\Gamma\vdash A^{\alpha}}=\ptol{\Gamma},\ptol{\alpha}\vdash \overline{\alpha}:A
\end{equation} 
\rosso{where for $\ptol{\Gamma}$ we set}
$$\ptol{\emptyset}=\emptyset\qquad\ptol{\Delta, B^{\beta}}=\ptol{\Delta},\ptol{\beta}, \overline{\beta}:B$$

We may now see how our $\NN$ rules appear under this translation. 
This will show that introduction and elimination rules for modal quantifiers in $\NN$ \emph{de facto} absorb explicit structural rules. We start with the total systems of Section~\ref{Sect-Total-Rules}. The axiom is immediately translated as
\[
\ptol{A^{\alpha}\vdash A^{\alpha}}=\ptol{\alpha},\overline{\alpha}:A\vdash\overline{\alpha}:A
\]
\rosso{while rules $\Box I$ and $\Box E$  are rewritten as follows: 
}

\begin{equation}\label{Eq:translatedrules}
\urule{\ptol{\Gamma}, \ptol{\alpha}, \overline{\alpha}\rl \overline{\alpha x}\vdash  \overline{\alpha x}:A }{\ptol{\Gamma}, \ptol{\alpha}\vdash\overline{\alpha}:\Box A }{}
\qquad
\urule{\ptol{\Gamma}, \ptol{\alpha} \vdash  \overline{\alpha}: \Box A }{\ptol{\Gamma}, \ptol{\alpha}, \overline{\alpha}\rl \overline{\alpha x}\vdash\overline{\alpha x}:A }{}
\end{equation}
\rosso{
which are valid instances of the corresponding rules in labelled systems. 
}

\rosso{
Let us now see how our proof of the formula {\bf D} (Proposition ~\ref{Prop:tot:weakcompl}, item (\ref{Prop:tot:weakcompl:D})) is converted into the labelled system. We start with 
}
\def\da{
\prova{ % premessa
 {\overline{<>}:\Box A}\vdash {\overline{<>}:\Box A}
  }
  { % conclusione
{\overline{<>}:\Box A}, \overline{<>}\rl\overline{x} \vdash {\overline{x}: A}
  }
  { \Box E_L } %nome 
}

\def\db{
\prova{ % premessa
  \da
  }
  { % conclusione
 {\overline{<>}:\Box A}, \overline{<>}\rl\overline{x} \vdash {\overline{<>}: \Diamond A}
  }
  { \Diamond I_L } %nome 
 } 
$$
\db
$$
but now we are stuck because we need to eliminate the relational formula 
$\overline{<>}\rl\overline{x}$ from the context of the conclusion. Here is where the additional (structural) rules on $\rl$ get into play. The labelled system for logic \D{} includes the rule (\textit{Ser}), which allows concluding:

\def\dc{
\prova{ % premessa
  {\overline{<>}:\Box A}, \overline{<>}\rl\overline{x} \vdash {\overline{<>}: \Diamond A}			
  }
  { % conclusione
   {\overline{<>}:\Box A} \vdash {\overline{<>}: \Diamond A}
  }
  {\textit{Ser}} %nome 
}

\def\de{
\prova{ % premessa
  \dc			
  }
  {  \vdash  {\overline{<>}:\Box A}\to{\overline{<>}: \Diamond A}  }
  {\to I} %nome 
}
$$
\de
$$

\rosso{
Analogously, we may see how our proof of the formula {\bf T} (Proposition~\ref{Prop:tot:weakcompl}, item (\ref{Prop:tot:weakcompl:T})) is translated, and amended with the explicit rule for reflexivity of $\rl$, (\textit{Refl}):
}

\def\da{
\prova{ % premessa
 {\overline{<>}:\Box A}\vdash {\overline{<>}:\Box A}
  }
  { % conclusione
{\overline{<>}:\Box A}, \overline{<>}\rl\overline{<>} \vdash {\overline{<>}: A}
  }
  { \Box E_L } %nome 
}

\def\db{
\prova{ % premessa
  \da
  }
  { % conclusione
 {\overline{<>}:\Box A} \vdash {\overline{<>}:  A}
  }
  {\textit{Refl} } %nome 
 } 

\def\dc{
\prova{ % premessa
  \db			
  }
  { % conclusione
   \vdash \overline{<>}: \Box A\to A}
  {\to I } %nome 
}

$$
\dc
$$

\rosso{
Finally, for formula {\bf 4} we need transitivity of the relation $\rl$.  The derivation in Proposition~\ref{Prop:tot:weakcompl}, item (\ref{Prop:tot:weakcompl:4}) is translated and amended using rule  (\textit{Trans}) as follows:
}

\def\da{
\prova{ % premessa
 {\overline{<>}:\Box A}\vdash {\overline{<>}:\Box A}
  }
  { % conclusione
{\overline{<>}:\Box A}, \overline{<>}\rl\overline{xy} \vdash {\overline{xy}: A}
  }
  { \Box E_L } %nome 
}
\def\db{
\prova{ % premessa
  \da
  }
  { % conclusione
 {\overline{<>}:\Box A}, \overline{<>}\rl\overline{x},\overline{x}\rl\overline{xy} \vdash {\overline{xy}:  A}
  }
  { \textit{Trans} } %nome 
 } 
\def\dc{
\prova{ % premessa
  \db			
  }
  { % conclusione
   {\overline{<>}:\Box A}, \overline{<>}\rl\overline{x}\vdash {\overline{x}: \Box A}
  }
  { \Box I_L} %nome 
}
\def\dd{
\prova{ % premessa
  \dc			
  }
  {%conclusione
   {\overline{<>}:\Box A}\vdash \overline{<>}: \Box\Box A} 
  {\Box I_L} %nome 
}

\def\de{
\prova{ % premessa
  \dd			
  }
  {%conclusione
   \vdash {\overline{<>}:\Box A\to \Box\Box A} }
  {\to I} %nome 
}
$$
\de
$$

\rosso{
The translation of total systems highlights well another perspective of the systemic difference between the labels induced by the translation of our $\NN$ and the labelled systems in the literature. 
Note that $\ptol{\alpha}$ (used only at the left of $\vdash$) introduces a sequence of relational formulas asserting that $\overline{\alpha}$ is reachable from $\overline{<>}$ via all the prefixes of the position $\alpha$, something that is not required in Negri/Simpson's systems. Indeed, differently from \cite{Simpson93,Negri:2011}, our labels have a tree-like structure, which permits the manipulation of subsequences of positions. This is why we do not give a formal translation of our $\NN$ into a known labelled framework (which is out of the scope of the paper). Our goal here is to show that our framework, once the structure and manipulation of positions are made explicit via labels and relational formulas, naturally produces labelled system \emph{similar} to those in~\cite{Simpson93,Negri:2011}. }

The dissimilarity becomes more evident for partial logics. 
The most glaring difference is the presence, in $\NN$, of  the existence predicate $E(\cdot)$, that parametrizes our modal rules in analogy with first-order frameworks. 
Since the assumption on the existence of labels does not hold for partial systems, we need to redefine the map on judgments, taking also into account existence formulas $E(\cdot)$.
To treat the logics of Section~\ref{sec:partlogics}, we thus define $\ptol{\cdot}$ on existence predicates as
$$
 \ptol{\ext{\alpha x}}=\overline{\alpha}\rl\overline{\alpha x}
$$
\rosso{and {modify} $\ptol{\cdot}$ on sequences of formulas as}

$$
\ptol{\emptyset}=\emptyset\qquad\ptol{\Delta,B^\beta}=\ptol{\Delta}, \overline{\beta}:B
$$

Consequently, we redefine $\ptol{\cdot}$ on  judgements (Equation~\ref{eq:traduzione}) in such a way that relational formulas are added only for labels that are predicated as existing:
\[
\qquad\qquad\ptol{\Gamma\vdash A^{\alpha}}=\ptol{\Gamma}\vdash \overline{\alpha}:A
\]

Rules $\Box I$ and $\Box E$ from Section~\ref{sec:partlogics} should be given in the explicitly labelled system as:

\[
\urule{\ptol{\Gamma}, \overline{\alpha}\rl \overline{\alpha x}\vdash  \overline{\alpha x}:A }{\ptol{\Gamma}\vdash\overline{\alpha}:\Box A }{}
\qquad
\urule{\ptol{\Gamma} \green{} \vdash  \overline{\alpha}: \Box A }{\ptol{\Gamma}, \overline{\alpha}\rl \overline{\alpha x}\vdash\overline{\alpha x}:A }{}
\]

As a single example, consider now the derivation of formula \textbf{K} that we gave in Proposition~\ref{prop:partial:weakcomp}, item~(\ref{prop:partial:weakcomp:K}), and which, for the sake of the reader, we reproduce here in linear form:

\def\auno{\prova{\pf{\Box A}{{<>}}\vdash \pf{\Box A}{{<>}}}{\pf{\Box A}{{<>}}, \ext{x}\vdash \pf{A}{{x}}}{\Box E}}
		\def\adue{\prova{\pf{\Box (A\to B)}{<>} \vdash \pf{\Box (A\to B)}{<>} }{\pf{\Box (A\to B)}{<>}, \ext{x} \vdash \pf{A\to B}{x}}{\Box E}}
		\def\atre{\prova{\auno\ \adue}{\pf{\Box A}{<>},\pf{\Box (A\to B)}{{<>}}, \ext{x} \vdash \pf{B}{x}}{\to E}}
		\def\aquattro{\prova{\atre}{\pf{\Box A}{<>},\pf{\Box (A\to B)}{{<>}} \vdash \pf{\Box B}{<>}}{\Box I}}
$$\aquattro$$
With the given rules it can be directly translated into the labelled system as follows:
  
\def\auno{\prova{\overline{<>}:\Box A \vdash \overline{<>}:\Box A }{\overline{<>}:\Box A, \overline{<>}\rl\overline{x} \vdash \overline{x}: A}{\Box E_L}}
		\def\adue{\prova{\overline{<>}:\Box (A\to B) \vdash \overline{<>}:\Box (A\to B) }{\overline{<>}:\Box (A\to B), \overline{<>}\rl\overline{x} \vdash \overline{x}: (A\to B)}{\Box E_L}}
		\def\atre{\prova{\auno\ \adue}{\overline{<>}:\Box A,\overline{<>}:\Box (A\to B), \overline{<>}\rl\overline{x} \vdash \overline{x}:  B}{\to E}}
		\def\aquattro{\prova{\atre}{\overline{<>}:\Box A,\overline{<>}:\Box (A\to B) \vdash \overline{<>}: \Box B}{\Box I_L}}
			$$\aquattro$$

\rosso{
Total systems, therefore, could be seen as partial ones where $E(\alpha)$ is implicitly assumed for any positions, which is consistent with the intended semantics of Section~\ref{sec:soundpartiallog}. }

\subsection{Future Work}\label{sec:conclusions}

Our investigation is open to different directions.  First, we plan to define and study the lambda-calculi that naturally emerge by making explicit the proof-term decoration in the systems of Section~\ref{sec:normalization}. Differently from the calculi presented, \textit{e.g.}, in~\cite{MM:ComInt:95} (where we did not have a sufficiently general notion of position), we will construct calculi where positions are first-class terms (and not mere decorations of terms), in such a way that positions could be manipulated by other lambda-terms. When formulated in a typed setting, this seems to require some notion of dependent types.

Moreover, we aim to study the 2-Sequent counterpart of the framework we presented here. As shown by~\cite{LellmannP19}, the definition of modular, analytic, and {cut-free} proof systems able to uniformly treat families of logics, is still an interesting problem. 
We claim that our notion of position allows pursuing a full notion of modularity: all logics (both total and partial) can be treated with the same set of rules, by simply tuning constraint on structural and modal rules and preserving cut-elimination.

\bigskip

\bibliographystyle{acm}
\bibliography{biblio} 

\begin{thebibliography}{10}

\bibitem{BaazIemhoff2005}
{\sc Baaz, M., and Iemhoff, R.}
\newblock On the proof theory of the existence predicate.
\newblock In {\em We Will Show Them! Essays in Honour of Dov Gabbay\/} (London,
  {U.K.}, 2005), S.~N. Art{\"e}mov, H.~Barringer, A.~S. d'Avila Garcez, L.~C.
  Lamb, and J.~Woods, Eds., College Publications, pp.~125--166.

\bibitem{BarMas:apal}
{\sc Baratella, S., and Masini, A.}
\newblock A proof-theoretic investigation of a logic of positions.
\newblock {\em Ann. Pure Appl. Logic 123}, 1-3 (2003), 135--162.

\bibitem{Baratella2004a}
{\sc Baratella, S., and Masini, A.}
\newblock {An infinitary variant of Metric Temporal Logic over dense time
  domains}.
\newblock {\em Mathematical Logic Quarterly 50}, 3 (2004), 249--257.

\bibitem{BaMaAML2004}
{\sc Baratella, S., and Masini, A.}
\newblock An approach to infinitary temporal proof theory.
\newblock {\em Arch. Math. Log. 43}, 8 (2004), 965--990.

\bibitem{BaMaJANCL13}
{\sc Baratella, S., and Masini, A.}
\newblock A natural deduction system for bundled branching time logic.
\newblock {\em Journal of Applied Non-Classical Logics 23}, 3 (2013), 268--283.

\bibitem{Baratella2019}
{\sc Baratella, S., and Masini, A.}
\newblock {A two-dimensional metric temporal logic}.
\newblock {\em Mathematical Logic Quarterly 13\/} (2019), 1--13.

\bibitem{Girard:ptlc}
{\sc Girard, J.-Y.}
\newblock {\em Proof theory and logical complexity}, vol.~1 of {\em Studies in
  Proof Theory. Monographs}.
\newblock Bibliopolis, Naples, 1987.

\bibitem{GMM:an-exp-ext-seq}
{\sc Guerrini, S., Martini, S., and Masini, A.}
\newblock An analysis of (linear) exponentials based on extended sequents.
\newblock {\em Logic Journal of the {IGPL} 6}, 5 (1998), 735--753.

\bibitem{GueMarMasPNGC2001}
{\sc Guerrini, S., Martini, S., and Masini, A.}
\newblock Proof nets, garbage, and computations.
\newblock {\em Theor. Comput. Sci. 253}, 2 (2001), 185--237.

\bibitem{GuMarMasCoher2003}
{\sc Guerrini, S., Martini, S., and Masini, A.}
\newblock Coherence for sharing proof-nets.
\newblock {\em Theor. Comput. Sci. 294}, 3 (2003), 379--409.

\bibitem{Kojima2012}
{\sc Kojima, K.}
\newblock {\em Semantical study of intuitionistic modal logics}.
\newblock PhD thesis, Graduate School of Informatics, Kyoto University, January
  2012.

\bibitem{Lellmann2015}
{\sc Lellmann, B.}
\newblock Linear nested sequents, 2-sequents and hypersequents.
\newblock In {\em Automated Reasoning with Analytic Tableaux and Related
  Methods - {TABLEAUX} 2015\/} (Cham, 2015), H.~De~Nivelle, Ed., vol.~9323 of
  {\em LNCS}, Springer International Publishing, pp.~135--150.

\bibitem{Lellmann2019}
{\sc Lellmann, B., and Pimentel, E.}
\newblock {Modularisation of Sequent Calculi for Normal and Non-normal
  Modalities}.
\newblock {\em ACM Transactions on Computational Logic 20}, 2 (apr 2019),
  1--46.

\bibitem{LellmannP19}
{\sc Lellmann, B., and Pimentel, E.}
\newblock Modularisation of sequent calculi for normal and non-normal
  modalities.
\newblock {\em {ACM} Trans. Comput. Log. 20}, 2 (2019), 7:1--7:46.

\bibitem{MM:ONFineStr:95}
{\sc Martini, S., and Masini, A.}
\newblock On the fine structure of the exponential rule.
\newblock In {\em Advances in Linear Logic\/} (Cambridge, {U.K.}, 1993), J.-Y.
  Girard, Y.~Lafont, and L.~Regnier, Eds., Cambridge University Press,
  pp.~197--210.

\bibitem{MM:ComInt:95}
{\sc Martini, S., and Masini, A.}
\newblock A computational interpretation of modal proofs.
\newblock In {\em Proof Theory of Modal Logics}, H.~Wansing, Ed. Kluwer,
  Dordrecht, 1996, pp.~213--241.

\bibitem{Mas:TwoSeqProof:92}
{\sc Masini, A.}
\newblock 2-sequent calculus: {A} proof theory of modalities.
\newblock {\em Ann. Pure Appl. Logic 58}, 3 (1992), 229--246.

\bibitem{Mas:TwoSeqInt:93}
{\sc Masini, A.}
\newblock 2-sequent calculus: Intuitionism and natural deduction.
\newblock {\em J. Log. Comput. 3}, 5 (1993), 533--562.

\bibitem{MVVJANCL2101}
{\sc Masini, A., Vigan{\`{o}}, L., and Volpe, M.}
\newblock Back from the future.
\newblock {\em Journal of Applied Non-Classical Logics 20}, 3 (2010), 241--277.

\bibitem{MVVENTCS2010}
{\sc Masini, A., Vigan{\`{o}}, L., and Volpe, M.}
\newblock A history of until.
\newblock {\em Electr. Notes Theor. Comput. Sci. 262\/} (2010), 189--204.

\bibitem{MVVJLC2011}
{\sc Masini, A., Vigan{\`{o}}, L., and Volpe, M.}
\newblock Labelled natural deduction for a bundled branching temporal logic.
\newblock {\em J. Log. Comput. 21}, 6 (2011), 1093--1163.

\bibitem{MasiniViganoZorzi08}
{\sc Masini, A., Vigan\`o, L., and Zorzi, M.}
\newblock {A Qualitative Modal Representation of Quantum Register
  Transformations}.
\newblock In {\em Proceedings of the 38th IEEE International Symposium on
  Multiple-Valued Logic (ISMVL 2008)}, G.~Dueck, Ed. IEEE Computer Society
  Press, Piscataway, NJ, 2008, pp.~131--137.

\bibitem{MVZ-jmvl}
{\sc Masini, A., Vigan{\`o}, L., and Zorzi, M.}
\newblock Modal deduction systems for quantum state transformations.
\newblock {\em J. Mult.-Valued Logic Soft Comput. 17}, 5-6 (2011), 475--519.

\bibitem{Negri:2011}
{\sc Negri, S.}
\newblock Proof theory for modal logic.
\newblock {\em Philosophy Compass 6}, 8 (2011), 523--538.

\bibitem{Pimentel2019}
{\sc Pimentel, E., Ramanayake, R., and Lellmann, B.}
\newblock {Sequentialising Nested Systems}.
\newblock In {\em Automated Reasoning with Analytic Tableaux and Related
  Methods. TABLEAUX 2019\/} (Cham, 2019), S.~Cerrito and A.~Popescu, Eds.,
  vol.~11714 of {\em LNCS}, Springer International Publishing, pp.~147--165.

\bibitem{Prawitz:1965}
{\sc Prawitz, D.}
\newblock {\em Natural deduction. {A} proof-theoretical study}.
\newblock Acta Universitatis Stockholmiensis. Stockholm Studies in Philosophy,
  No. 3. Almqvist \& Wiksell, Stockholm, 1965.

\bibitem{Scott1979}
{\sc Scott, D.}
\newblock Identity and existence in intuitionistic logic.
\newblock In {\em Applications of Sheaves: Proceedings of the Research
  Symposium on Applications of Sheaf Theory to Logic, Algebra, and Analysis,
  Durham, July 9--21, 1977\/} (Berlin, Heidelberg, 1979), M.~P. Fourman, C.~J.
  Mulvey, and D.~S. Scott, Eds., vol.~753 of {\em Lect. Notes Math.}, Springer,
  pp.~660--696.

\bibitem{Simpson93}
{\sc Simpson, A.}
\newblock {\em The proof theory and semantics of intuitionistic modal logic}.
\newblock PhD thesis, University of Edinburgh, UK, 1993.

\bibitem{tvd1988}
{\sc Troelstra, A.~S., and van Dalen, D.}
\newblock {\em Constructivism in Mathematics (2 volumes)}.
\newblock North-Holland, Amsterdam, 1988.

\bibitem{Vigano00a}
{\sc Vigan{\`o}, L.}
\newblock {\em {Labelled Non-Classical Logics}}.
\newblock Kluwer Academic Publishers, Dordrecht, 2000.

\end{thebibliography}
\end{document}